\documentclass[11pt,letterpaper]{article}
\usepackage[margin=1in]{geometry}
\usepackage{amsmath,amssymb}
\usepackage{changepage}
\usepackage[utf8x]{inputenc}
\usepackage{textcomp,marvosym}
\usepackage{cite}
\usepackage{nameref,hyperref}
\usepackage[right]{lineno}

\usepackage{microtype}
\DisableLigatures[f]{encoding = *, family = * }
\usepackage[table]{xcolor}
\usepackage{array}
\newcolumntype{+}{!{\vrule width 2pt}}

\newlength\savedwidth

\usepackage[aboveskip=1pt,labelfont=bf,labelsep=period,justification=raggedright,singlelinecheck=off]{caption}

\makeatletter
\renewcommand{\@biblabel}[1]{\quad#1.}
\makeatother

\usepackage{lastpage,fancyhdr,graphicx}
\usepackage{epstopdf}

\usepackage{algorithm}
\usepackage{algpseudocode}
\usepackage{graphicx} 

\usepackage{amsmath,amsfonts,bm}
\usepackage{bbm}

\newcommand*{\rom}[1]{%
\textup{\uppercase\expandafter{\romannumeral#1}}%
}

\def\1{\bm{1}}

\DeclareMathAlphabet{\mathsfit}{\encodingdefault}{\sfdefault}{m}{sl}
\SetMathAlphabet{\mathsfit}{bold}{\encodingdefault}{\sfdefault}{bx}{n}

\DeclareMathOperator*{\argmin}{argmin}

\DeclareMathOperator{\Tr}{Tr}

\newcommand{\R}{\mathbb{R}}

\usepackage{amsthm}
\newtheorem{theorem}{Theorem}[section]

\newcommand{\norm}[1]{\left\lVert#1\right\rVert}
\usepackage{comment}
\usepackage{xcolor}
\usepackage{textcomp}

\newcommand{\TP}{\mathrm{TP}}
\newcommand{\FP}{\mathrm{FP}}
\newcommand{\FN}{\mathrm{FN}}
\newcommand{\precision}{\mathrm{Precision}}
\newcommand{\recall}{\mathrm{Recall}}

\usepackage{textcomp} 
\usepackage{authblk}

\begin{document}

\vspace*{0.2in}

{\huge
\textbf\newline{Efficient inference of dynamic gene regulatory networks using discrete penalty} 
}

\begin{flushleft}

Visweswaran Ravikumar\textsuperscript{1}\textsuperscript{\textdagger},
Aaresh Bhathena\textsuperscript{2\textsuperscript{\textdagger}},
Wajd N Al-Holou\textsuperscript{3},
Salar Fattahi\textsuperscript{2},
Arvind Rao\textsuperscript{1,4,5}
\\
\bigskip
\textsuperscript{1} Department of Bioinformatics and Computational Medicine, University of Michigan, Ann Arbor, MI, USA
\\
\textsuperscript{2} Department of Industrial and Operations Engineering, University of Michigan, Ann Arbor, MI, USA
\\
\textsuperscript{3} Department of Neurosurgery, University of Michigan, Ann Arbor, MI, USA
\\
\textsuperscript{4} Department of Biostatistics, University of Michigan, Ann Arbor, MI, USA
\\
\textsuperscript{5} Department of Radiation Oncology, University of Michigan, Ann Arbor, MI, USA
\\
\bigskip

$\dagger$ These authors contributed equally to this work.

\end{flushleft}

\section*{Abstract}
Gene regulatory networks (GRNs) orchestrate cellular decision making and survival strategies. Inferring the structure of these networks from high-dimensional transcriptomics data is a central challenge in systems biology. Traditional approaches to GRN inference, such as the graphical lasso and its joint extensions, rely on $\ell_1$ penalty to induce sparsity but can bias network recovery and require extensive hyperparameter tuning. Here, we present a scalable framework for the joint inference of dynamic GRNs using a discrete $\ell_0$ penalty, enabling direct and unbiased control over network sparsity. Leveraging recent algorithmic advances, we efficiently solve the resulting mixed-integer optimization problem for populations structured as arbitrary tree hypergraphs, accommodating both continuous and categorical distinctions among biological samples. After validating our method on synthetic benchmarks, we apply it to single-cell and spatial transcriptomics data from glioblastoma (GBM) patient tumors. Our approach reconstructs gene networks across tumor clusters, maps network rewiring along hypoxia gradients, and reveals niche-specific differences between primary and recurrent tumors. By providing a robust and interpretable tool for GRN inference in complex tissues, our work facilitates high-resolution dissection of tumor heterogeneity and adaptation, with broad applicability to emerging large-scale transcriptomic datasets.

\section*{Author summary}  

Efficient inference of GRNs from gene-expression measurements is a cornerstone problem in systems biology. Reconstructing the regulatory logic of cells from gene-expression data allows us to identify central drivers of cellular adaptation to specific environmental conditions or along disease stages. This is particularly useful when applied to high-throughput single cell and spatial transcriptomics (ST) datasets where we can map continuous changes to the regulatory networks at a high resolution in physical or latent space. However, the size and complexity of the optimization problems often necessitate use of approximation techniques that trade statistical accuracy for computational efficiency. In this work, we introduce a new algorithm that leverages recent mathematical advances to preserve statistical accuracy and minimize bias in estimation of interaction strengths. As a result, our approach more accurately identifies gene interactions without compromising computational speed. We evaluated our algorithm with extensive numerical simulations, and show applications to single cell and ST datasets in glioblastoma. Our algorithm is readily extensible to high-resolution ST data, is highly scalable and versatile, and can be adopted for studying a variety of biological phenomena. We make our developed algorithms open-source, and believe it will be a valuable contribution to the research community.

\section{Introduction}

Gene regulatory networks (GRNs) are central to cellular decision-making, encoding context-specific responses to developmental cues and microenvironmental changes~\cite{barabasi2011network}. The advent of high-throughput single-cell and spatial transcriptomic technologies has enabled unprecedented resolution in profiling gene expression across developmental, spatial, and disease contexts~\cite{gulati2025profiling,desta2025advancements}. Understanding the dynamics of GRNs---rather than isolated gene expression changes---provides deeper insight into cellular state transitions, regulatory hierarchies, and potential therapeutic targets, especially in complex diseases such as cancer~\cite{trapnell2014dynamics,dimitrakopoulos2018network}. A fundamental challenge in GRN inference is reducing the number of false positive discoveries, as biological networks are inherently sparse, with only a subset of possible gene-gene interactions realized in any given context~\cite{lundqvist2025topology,frontiers2022sparsity}.

A principled way to promote the desired sparsity during network inference is through the use of the $\ell_0$ penalty function~\cite{marjanovic2015l_}. This function directly penalizes the number of non-zero edges of the network, thereby favoring simple and interpretable network structures. However, incorporating the $\ell_0$ penalty leads to a non-convex formulation that is computationally intractable at the scale required for realistic networks. To address this challenge, convex relaxations based on the $\ell_1$ penalty are often employed. The $\ell_1$ penalty promotes sparsity by penalizing the absolute values of edge weights, effectively removing weaker connections. Nevertheless, it applies uniform shrinkage across all edge weights, introducing bias in the estimates of stronger interactions (see~\cite[Example 1]{fattahi2021scalable}). This drawback underlies limitations of the popular \textit{graphical lasso}(GL)~\cite{friedman2008sparse} and its extension~\cite{danaher2014joint,hallac2017network}. In contrast, exact solutions to the original $\ell_0$ penalized formulation circumvent these issues~\cite{stromberg1992breakdown,rousseeuw2005robust,zioutas2005deleting,xu2024integer,kucukyavuz2023consistent,manzour2021dag,kim2021scalable}. Motivated by this gap and aiming to address it, we introduce an algorithm that leverages the $\ell_0$ penalty to directly enforce sparsity without inducing shrinkage bias, thereby enabling more accurate inference of gene regulatory networks.

This work presents an efficient and scalable inference procedure that addresses the computational challenges introduced by the $\ell_0$ penalty. The algorithm offers two distinct advantages. First, the optimization problem decomposes across genes, enabling parallel computation and making the method suitable for high-dimensional settings. In our numerical experiments, we scale to 2,000 genes, demonstrating its practicality for real-world inference tasks. Second, although each subproblem is non-convex due to the $\ell_0$ penalty, we employ a specialized algorithm~\cite[Algorithm 2]{bhathena2025parametric} based on dynamic programming, which solves each subproblem in quadratic time with respect to the number of populations---providing a crucial advantage when analyzing high-resolution datasets to study fine-grained changes in regulatory logic of cells along spatio-temporal gradients. Our experiments show that the runtime of the proposed approach is comparable to state-of-the-art methods based on convex $\ell_1$ penalty, and significantly faster than classical techniques such as the graphical lasso, which directly optimize the penalized log-likelihood.

We further generalize our method to handle populations with explicit categorical distinctions, such as primary versus recurrent tumor samples in defined spatial niches. After rigorous validation on synthetic benchmarks, we demonstrate the utility of our approach in three key applications using single-cell and spatial transcriptomics data from glioblastoma that was curated locally by our team from patients in the Michigan Medicine hospital: 
(i) reconstructing gene networks across tumor clusters in a single-cell GBM atlas, 
(ii) mapping continuous network rewiring along hypoxia gradients in GBM spatial tissue sections, and 
(iii) identifying niche-specific network differences between primary and recurrent tumors. As large-scale transcriptomic datasets become increasingly prevalent, our work provides a scalable and accurate computational framework for dissecting tumor biology at high resolution.

 \paragraph{Notation} We use lowercase letters to denote scalars and vectors, and uppercase letters to denote matrices. The $(i,j){\text{th}}$ entry of a matrix $M$ is denoted by either $M_{i,j}$ or $[M]_{i,j}$.  Let $\mathbb{S}^p_+$ denote the set of $p \times p$ symmetric positive semi-definite matrices, that is, symmetric matrices whose eigenvalues are all non-negative. For a matrix $M$, we use $\|M\|_{\ell_q}$ to represent the element vise $\ell_{q}$-norm of $M$, defined as $\|M\|_{\ell_q}=\left(\sum_{i,j} \left|M_{i,j}\right|^q\right)^{1/q}$. The notation $\|M\|_{0,\mathrm{off}}$ refers to the total number of nonzero off-diagonal elements in $M$. The trace of $M$, denoted by $\operatorname{Tr}(M)$, is the sum of its diagonal elements, and $\det(M)$ denotes the determinant of $M$. The scalar sign function, denoted by $\operatorname{sign}(\cdot)$, is defined as $\operatorname{sign}(x) = |x|/x$ for $x \neq 0$, where $|x|$ denotes the absolute value of $x$. Given two functions $f(n)$ and $g(n)$, we write $f(n)=\mathcal{O}(g(n))$ when there exists a universal constant $C>0$ satisfying $f(n) \leq Cg(n)$ for all large enough $n$. For a function $f:\mathcal{X}\to\R$ the notation $\argmin_{x\in \mathcal{X}} f(x)$ denotes a global minimizer of $f$; that is, if $x^*=\argmin_{x\in\mathcal{X}} f(x)$ then $f(x^*)\le f(x)$ for all $x\in \mathcal{X}$. When the domain $\mathcal{X}$ is clear from context, we sometimes omit it for brevity.

\subsection{Background and related work} 
In this work, we consider the problem of jointly inferring gene network structure for several related populations that share varying levels of proximity to each other. This similarity may arise from experimental design (e.g., samples collected at different time points or spatial locations) or may be data-driven (e.g., proximity between single-cell clusters in a latent embedding). The populations could represent clusters of cells aligned along a (pseudo-)temporal axis in a single cell dataset, or spatially-localized clusters in spatial transcriptomics (ST) slides. A widely used framework to model such dynamic networks is based on \textit{Markov random fields} (MRFs). Specifically, each population $k$ is associated with an undirected graph $G_k(V_k, E_k)$, where $V_k$ represents the set of random variables corresponding to genes, and $E_k$ encodes the conditional dependence among them. When the underlying distribution of the random variables is assumed to be Gaussian, the model is referred to as a Gaussian Markov Random Field (GMRF). In this case, the structure of the graph $(G_k)$ is fully captured by the sparsity of the precision matrix $\Theta_k$ (the inverse of the covariance matrix). Specifically, a zero value for the $(i,j){\text{th}}$ entry $[\Theta_k]_{i,j}$ implies that variables $i$ and $j$ are conditionally independent, given all other variables at population $k$. As such, accurate estimation of the precision matrices $\{\Theta_k\}$ is fundamental to recovering the underlying gene regulatory networks and understanding their variations across populations.

For a single population, the graphical lasso~\cite{friedman2008sparse} estimates the precision matrix $ \Theta$ by minimizing the penalized negative likelihood. The resulting estimator, denoted by $\hat \Theta$, is given by:
\begin{align}\tag{MLE}\label{eq:MLE}
    \hat \Theta = \argmin_{\Theta \in \mathbb{S}^p_+} \left\{ \operatorname{Tr}(\Theta \hat\Sigma) - \log\det(\Theta) + \lambda \|\Theta\|_{\ell_1} \right\},
\end{align}
where $\hat\Sigma$ denotes the sample covariance of the population and the scalar $\lambda\ge 0$ controls the degree of sparsity. To capture dependencies across multiple populations, the Joint Graphical Lasso (JGL)~\cite{danaher2014joint,ma2016joint} extends this framework by adding a similarity penalty between the inferred networks:
\begin{align}\label{eq:JGL}
    \{\hat \Theta_k\}_{k=1}^K =
   \argmin_{\Theta_1, \ldots, \Theta_K \in \mathbb{S}^p_+}\! \left\{  \sum_{k=1}^K \left( \operatorname{Tr}(\Theta_k \hat\Sigma_k) - \log\det(\Theta_k) \right) + \sum_{k=1}^K  \lambda \|\Theta_k\|_{\ell_1} + P(\Theta_1, \ldots, \Theta_K) \right\},\tag{JGL}
\end{align}
where individual populations are identified by their subscript $k$, and  $P(\cdot)$ enforces the desired level of similarity among the inferred networks.

While maximum likelihood-based approaches are widely adopted, they become computationally intensive for large-scale problems due to the need to optimize over the log-determinant term~\cite{pmlr-v54-wang17e,fattahi2021scalable}. Alternative estimators such as CLIME~\cite{cai2011constrained} and the Elementary Estimator~\cite{yang2014elementary} introduce reformulations that avoid the log-determinant term, improving scalability. These methods have been further extended to joint and time-varying settings~\cite{lee2015joint,fattahi2021scalable,fattahi2023solution}. However, they rely on the $\ell_1$ penalty, which introduces bias in the estimates and can limit recovery of the true network structure. In contrast, the $\ell_0$ penalty offers a direct and unbiased approach to enforcing sparsity. Incorporating the $\ell_0$ penalty leads to an optimization problem that falls within the class of mixed-integer quadratic programs (MIQPs). While solving general MIQPs is NP-hard~\cite{bertsimas2016best,bertsimas2020certifiably}, recent advances have identified tractable instances of these challenging problems~\cite{han2022polynomial,liu2023graph,gomez2024real}. 
In our recent work, we showed that when the dependency structure across populations forms a tree (i.e., an acyclic graph), the corresponding MIQP instance can be solved efficiently by ~\cite[Algorithm 2]{bhathena2025parametric}. Additional details on the algorithm are provided in Methods.

\subsection{Our contributions}
\begin{enumerate}
    \item \textbf{Sparsity control with $\ell_0$ penalty:} We extend our previous framework~\cite{ravikumar2023efficient} to incorporate the $\ell_0$ penalty for joint inference of multiple related GMRFs. By leveraging recent algorithmic advances~\cite{bhathena2025parametric}, we efficiently solve the resulting MIQP for tree-structured spatial relationships, enabling scalable and accurate inference of context-specific gene regulatory networks.
    \item \textbf{Handling categorical populations:} Our method accommodates explicit categorical distinctions between populations, such as primary versus recurrent tumor samples in distinct spatial niches.
    \item \textbf{Demonstrated utility in glioblastoma biology:} We validate our approach on synthetic benchmarks, and then showcase the application of our framework to three distinct scenarios in glioblastoma transcriptomics datasets: (i) network reconstruction across tumor clusters, (ii) mapping network rewiring along hypoxia gradients, and (iii) identifying niche-specific differences between primary and recurrent tumors.
\end{enumerate}

\section{Materials and methods}

In this section, we describe our framework for the joint inference of the precision matrices. We begin by presenting additional definitions used throughout the manuscript. Subsequently, we introduce an optimization problem for the joint inference of GMRFs and demonstrate how this formulation naturally decomposes, enabling parallel solution of subproblems. We then present the algorithm and provide a comprehensive analysis of its computational complexity. The following subsections describe an extension of our algorithm for categorical data and the implementation details of our algorithm.

Let $p$ denote the number of genes of the underlying network. For the $k$th population, gene expression data is collected in a matrix $X_k\in\R^{p\times n_k}$ where each column corresponds to a gene expression profile of a single biological sample and $n_k$ is the number of such samples. This work assumes that each population follows a multivariate Gaussian distribution with zero mean and covariance matrix denoted by $\Sigma_k^* \in \mathbb{S}^p_+$ for populations $k = 1, \ldots, K$. In other words, every column of $X_k$ is sampled from $\mathcal{N}(0, \Sigma_k^*)$. The zero mean assumption is common practice in single-cell analysis~\cite{wolf2018scanpy,hao2024dictionary} and follows without loss of generality, as the data can be centered and rescaled. The sample covariance matrix for the $k$th population is given by $\hat\Sigma_k=\frac{1}{n_k}X_{k}(X_{k})^\top$. Finally, we use $\Theta^*_k=(\Sigma^*_k)^{-1}$ to denote the true precision matrix and $\hat\Theta_k$ to denote the inferred precision matrix of the $k$th population. 

The central idea of our approach is to infer $\{\Theta^*_k\}_{k=1}^K$ from the observed sample covariances $\{\hat\Sigma_k\}_{k=1}^K$. To achieve this, we first introduce the concept of a population \textit{hypergraph} to model the neighborhood structure among different populations.

\paragraph{Population hypergraph} To capture dependencies among the $K$ populations, we introduce a second undirected graph (the \textit{Hypergraph}) $\mathcal{H}(V_H,W_H)$, whose vertices $V_H$ denotes the different populations and the edge weights capture some notion of proximity among them. The edge weights of the hypergraph, denoted $W_{i,j}$, control the strength of the similarity between populations $i$ and $j$, with larger weights encouraging greater alignment in the inferred networks. The matrix $W\in\R^{K\times K}$ collects all such edge weights and encodes the structure of the hypergraph. A key structural assumption we make is that the hypergraph $\mathcal{H}$ forms a tree structure --- i.e., it does not contain loops. This condition is generally met in biological systems, as most developmental and disease processes follow a branching, pseudo-temporal progression \cite{saelens2019comparison, trapnell2014dynamics}. It also facilitates solving our proposed optimization problem with an $\ell_0$ penalty, as we will detail later. As an illustrative example, Fig~\ref{fig:schema-direct} shows the tree \textit{hypergraph} connecting different spatial niches in a GBM ST slide, and changing edge-weights in this graph allows for flexible modeling of gene networks across these spatial environments.

\begin{figure}
    \centering
    \includegraphics[width=0.3\linewidth]{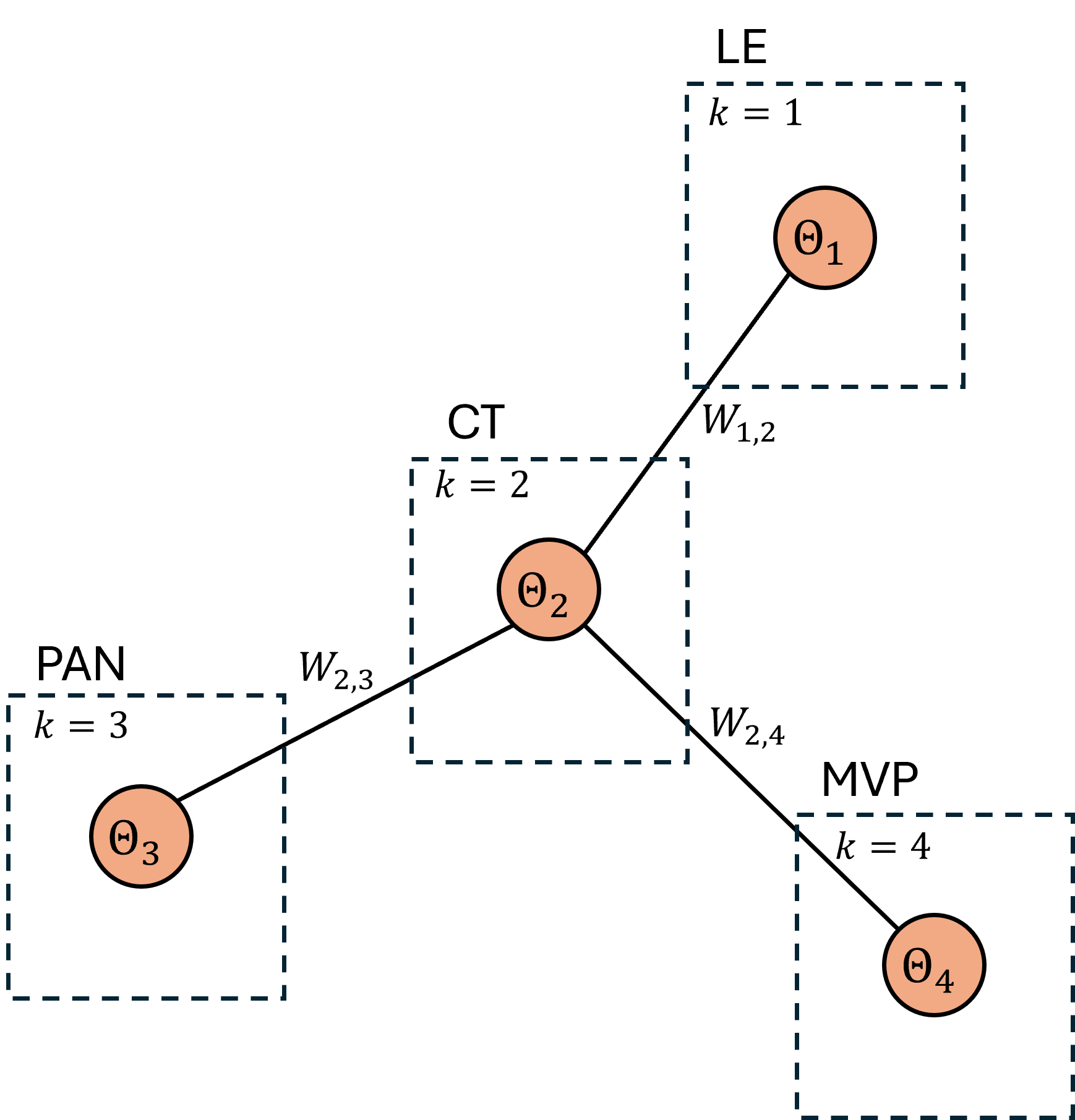}\qquad\qquad\qquad 
    \includegraphics[width=0.45\linewidth]{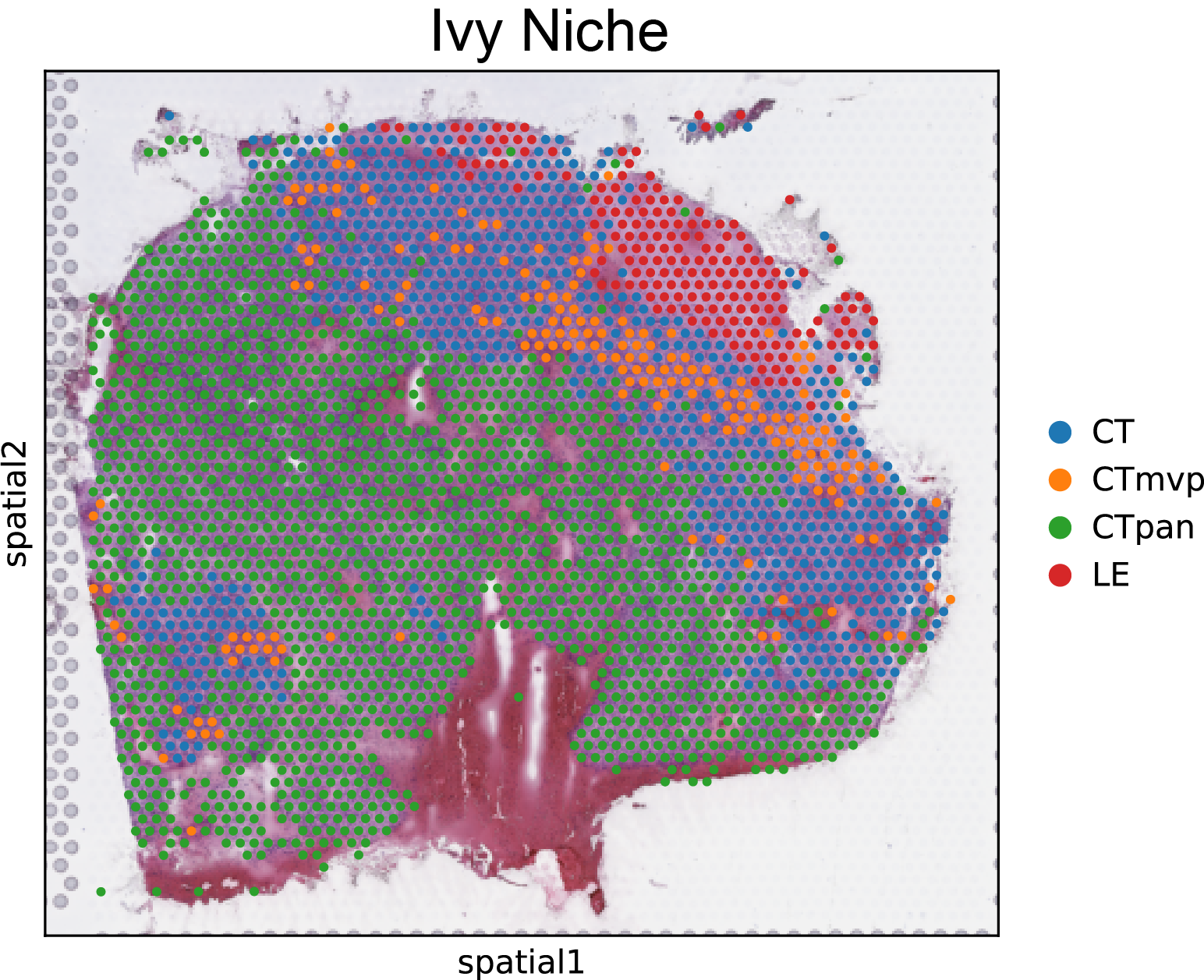}
    \caption{\textbf{Representative scenario for joint network inference across multiple spatial niches.} Spatial clusters (niches) are modeled as distinct populations. The relationships among populations is flexibly modeled using arbitrary tree-structured hypergraphs. Edge weights in the hypergraph encode similarity between populations, and guide the joint inference algorithm in estimation of shared and unique regulatory features. }
    \label{fig:schema-direct}
\end{figure}

\paragraph{Problem formulation and inference procedure}

We now present the optimization problem used for estimating the precision matrices $\hat \Theta_k$ from the sample covariances $\hat\Sigma_k$. Recall that the sample covariance matrices are computed from the gene expression data across different populations. The relative proximity among these populations is modeled using the hypergraph, with its structure encoded in matrix $W\in\R^{K\times K}$. The precision matrices $\{\hat \Theta_k\}_{k=1}^K$ are estimated by solving the following optimization problem: 

\begin{equation}\tag{ELEM-$0$}
	\begin{aligned}
	\{\hat{\Theta}_k\}_{k=1}^K =\argmin_{\Theta_1,\ldots,\Theta_k}  \underbrace{\sum_{k=1}^K\left \|\Theta_k - \tilde{F}^*(\hat{\Sigma}_{k})\right\|_{\ell_2}^2}_{\text{backward mapping deviation}} + \underbrace{\lambda\sum_{k=1}^K \left\|{\Theta_k}\right\|_{0,\mathrm{off}}}_{\ell_0\text{-penalty}}+\underbrace{\gamma \sum_{l>k}W_{kl}\left\|\Theta_{k}-\Theta_{l}\right\|_{\ell_2}^2}_{\text{similarity penalty}}.
    \end{aligned}
	\label{eq: l0}
\end{equation}
In the expression above, the term \( \tilde{F}^*(\hat{\Sigma}_k) \) denotes an approximation of the so-called \textit{backward mapping}, a quantity that plays a central role in the inference of MRFs. For further technical details on the (approximate) backward mapping and its connection to MRFs, we refer the reader to \cite[Chapter 3]{wainwright2008graphical}.
 Classical algorithms like JGL estimate $\{\hat{\Theta}_k\}_{k=1}^K $ using a variant of the backward mapping derived from the maximum likelihood principle. However, this approach requires optimizing a log-determinant term, rendering it computationally expensive and intractable at scale. To address this, we adopt the approximate backward mapping introduced by Yang et al.~\cite{yang2014elementary}. For a fixed constant \( \nu > 0 \), this approximate backward mapping is defined as the inverse of the soft-thresholded \( \hat{\Sigma}_k \), that is:
\begin{align}\label{eq: app backward map}
    \tilde F^*(\hat \Sigma_k)=[\texttt{ST}_\nu(\hat \Sigma_k)]^{-1},
\end{align}
where $\texttt{ST}_{\nu}:\R^{p\times p}\to \R^{p\times p}$ is the soft thresholding operator defined as 
\begin{align*}
    [\texttt{ST}_{\nu}(M)]_{i,j}=\begin{cases}
        M_{i,j}-\operatorname{sign}(M_{i,j})\cdot\min\{\nu,|M_{i,j}|\}\quad&\text{ when } i\ne j\\
        M_{i,j}&\text{otherwise}
    \end{cases}
\end{align*}
A key advantage of the above backward mapping lies in its computational efficiency. Specifically, after soft-thresholding, $\hat{\Sigma}_k$ becomes a sparse matrix, whose inverse can be computed efficiently using techniques such as sparse Cholesky decomposition.

In the formulation given in~\eqref{eq: l0}, the \textit{backward mapping deviation} term encourages $\hat\Theta_k$ to be close to the approximate backward mapping $\tilde F^*(\hat\Sigma_k)$ for $k=1,\ldots,K$. The $\ell_0$ penalty promotes sparsity by penalizing the number of non-zeros in the off-diagonal elements in $\{\hat\Theta_k\}_{k=1}^K$. The \textit{similarity penalty} enforces similarity between precision matrices of neighboring populations, with $W_{k,l}$ reflecting the relative proximity. The tradeoff between sparsity and similarity is controlled by parameters $\lambda,\gamma>0$, respectively. 

From a computational perspective, the formulation above has a distinct advantage: the optimization problem decomposes across every $(i,j)$ coordinate of the matrix variables $\{\Theta_k\}_{k=1}^K$. This allows the problem to be split into independent subproblems, one for each $(i,j)$ coordinate, which can be solved in parallel. Every subproblem optimizes over $K$ variables, corresponding to different populations. We handle the diagonal and off-diagonal coordinates separately. For the diagonal elements, the problem reduces to a \textit{quadratic program} given by
\begin{align}
    \hspace{-2mm}\{[\hat{\Theta}_k]_{i,i}\}_{k=1}^K =
    \argmin_{\{[\Theta_k]_{i,i}\}_{k=1}^K}   \sum_{k=1}^K\left([\Theta_k]_{i,i} - \left[\tilde{F}(\hat{\Sigma}_{k})\right]_{i,i}\right)^2 +\gamma \sum_{l>k}W_{kl}\left([\Theta_{k}]_{i,i}-[\Theta_{l}]_{i,i}\right)^2.
    \label{eq: diag elem l0}
\end{align}
This problem admits a closed-form solution. Moreover, since $W$ encodes a tree-structured hypergraph, the overall solution can be obtained in $\mathcal{O}(K)$ time~\cite{george1981computer,vandenberghe2015chordal}.

Due to the symmetry of the precision matrices, it suffices to consider only the off-diagonal entries corresponding to the index pair $(i,j)$ with $1\le i< j\le p$. For each such pair, the original problem decomposes into the following element-wise subproblem
\begin{align}
	\begin{aligned}
	\{[\hat{\Theta}_k]_{i,j}\}_{k=1}^K =
    \argmin_{\{[\Theta_k]_{i,j}\}_{k=1}^K}  \sum_{k=1}^K\left([\Theta_k]_{i,j} - \left[\tilde{F}(\hat{\Sigma}_{k})\right]_{i,j}\right)^2 + \lambda\sum_{k=1}^K \|{[\Theta_k]_{i,j}}\|_{0}+\gamma \sum_{l>k}W_{kl}\left([\Theta_{k}]_{i,j}-[\Theta_{l}]_{i,j}\right)^2.
    \end{aligned}
	\label{eq: elem l0 ij}   
\end{align}
To handle the $\ell_0$ penalty term, the optimization problem above can be reformulated as an MIQP. In the general case, solving MIQPs is NP-hard, with worst-case complexity that scales exponentially with the number of populations $K$, rendering the problem intractable for large-scale instances. When the spatial structure encoded by the matrix $W$ forms a tree, Bhathena et al.~\cite{bhathena2025parametric} showed that the problem can be solved in $\mathcal{O}(K^2)$ time via dynamic programming. Accordingly, to solve the optimization problem~\eqref{eq: elem l0 ij}, we adopt this dynamic programming approach. Its computational efficiency enables us to recover gene network changes at high resolution along spatio-temporal gradients in single-cell and spatial transcriptomics datasets, as demonstrated in our applications.

Building on the optimization problem, we introduce the algorithm for solving~\eqref{eq: l0}. A complete description is provided in Algorithm~\ref{alg}. The procedure begins by computing the sample covariance matrices $ \{\hat{\Sigma}_k\}_{k=1}^K$ from the gene expression data $\{X_k\}_{k=1}^K$. Next, the approximate backward mapping $\{\tilde{F}^*(\hat{\Sigma}_k)\}_{k=1}^K$ is computed according to Equation~\eqref{eq: app backward map}. For every index pair $1\le i\le j \le p$, we then apply the dynamic programming algorithm from~\cite{bhathena2025parametric} to estimate the corresponding entries $[\hat\Theta_k]_{i,j}$ for all $k$. Finally, the lower triangular coordinates are set equal to the corresponding upper triangular entries, i.e., $[\hat\Theta_k]_{j,i} \leftarrow [\hat\Theta_k]_{i,j}$ for $i<j$.

\begin{algorithm}[h]
		\caption{Algorithm to Solve~\eqref{eq: l0}}
		\textbf{Input:} Data samples $\{X_k\}_{k=1}^K$, parameters $(\mu,\gamma,\nu)$, and weight matrix $W$\\
		\textbf{Output:} Precision matrix $\{\hat \Theta_k\}_{k=1}^K$
		\begin{algorithmic}[1]
            \State Compute the sample covariance $\hat\Sigma_{k}=\frac{1}{n_k} X_kX_{k}^\top$ for every $k=1,\ldots,K$
            \State Compute the approximate backward mapping $\tilde F^*(\hat \Sigma_k)=\left[\texttt{ST}_{\nu}\left(\hat\Sigma_k\right)\right]^{-1}$
			\For{\text{ every } $1\le i\le j\le p$}
            \State Obtain $\left\{ [\hat\Theta_k]_{i,j} \right\}_{k=1}^K$ by solving Equation~\eqref{eq: elem l0 ij} using \cite[Algorithm 2]{bhathena2025parametric}
            \State $[\hat\Theta_k]_{j,i} \leftarrow [\hat\Theta_k]_{i,j}$ for $k=1,\ldots,K$ and $i\ne j$
            \EndFor
			\State\Return $\{\hat\Theta_k\}_{k=1}^K$
		\end{algorithmic}
        \label{alg}
\end{algorithm}

\begin{theorem}\label{thm}
    Algorithm~\ref{alg} has a runtime of $\mathcal{O}\left(Kp^2n+Kp^3+K^2p^2\right)$ and requires $\mathcal{O}\left(Kpn+K^2p^2\right) $ memory.
\end{theorem}
\begin{proof}
The algorithm's computational complexity is determined by three key components. First, computing the sample covariance matrix for each population involves matrix multiplication, which requires $\mathcal{O}(p^2 n)$ time and $\mathcal{O}(pn + p^2)$ memory per population. With $K$ populations, this step contributes $\mathcal{O}(Kp^2 n)$ time and $\mathcal{O}(Kpn + Kp^2)$ memory overall. 

Second, computing the approximate backward mapping $\tilde{F}^*(\hat{\Sigma}_k)$ for each population entails an element-wise soft-thresholding followed by a matrix inversion, which incurs $\mathcal{O}(p^3)$ time and $\mathcal{O}(p^2)$ memory per population, resulting in a total cost of $\mathcal{O}(Kp^3)$ time and $\mathcal{O}(Kp^2)$ memory.

Third, solving the core optimization subproblem for each index pair $(i,j)$ via dynamic programming has both time and memory complexity $\mathcal{O}(K^2)$~\cite[Theorem 2]{bhathena2025parametric}. Since this must be performed for all $\frac{p(p+1)}{2}$ matrix coordinates, the total cost of this step is $\mathcal{O}(K^2p^2)$ in both time and memory. 

Combining all components, the overall computational complexity of the algorithm is dominated by the sum of these contributions.
\end{proof}

\subsection{Inference over multiple categories}\label{sec: categorical ELEM-0}

\begin{figure}
    \centering
    \includegraphics[width=0.38\linewidth]{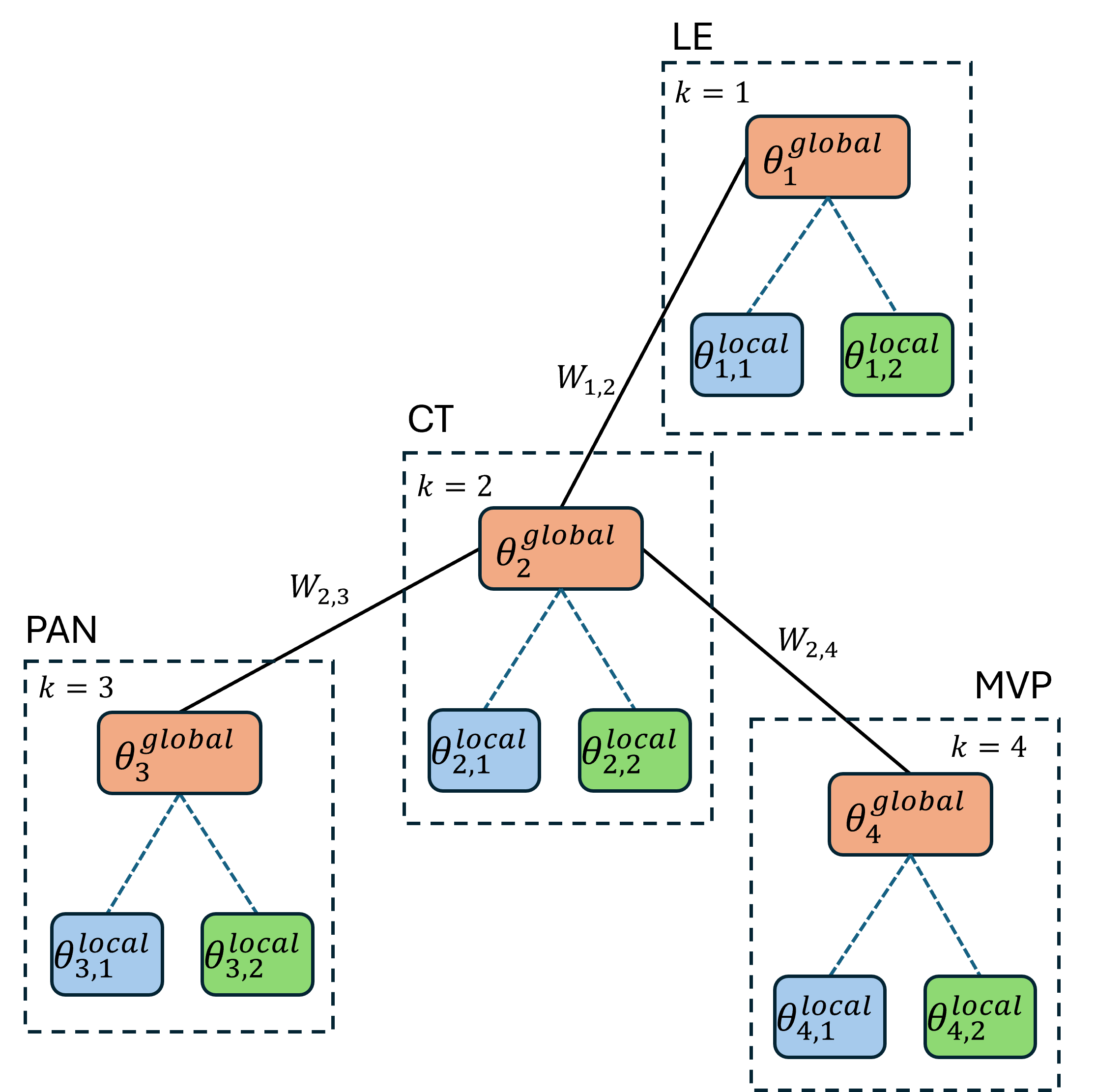}\qquad\qquad\qquad 
    \includegraphics[width=0.42\linewidth]{Figures/ivy-gap.png}
    \caption{\textbf{Schematic of condition-specific network inference on a global hypergraph framework.} Each node in the global hypergraph represents a distinct spatial cluster, with two associated conditions (e.g., primary and recurrent disease). The algorithm jointly infers gene regulatory network structures for each niche-condition pair by inferring both shared (global) and condition-specific edges. This approach enables the identification of niche-specific regulatory rewiring in recurrent tumors, revealing molecular programs that drive treatment resistance and aggressive tumor growth.}
    \label{fig:cat-schema}
\end{figure}

Estimating large-scale networks typically requires more observations than variables $(n \gg p)$ to ensure statistical reliability~\cite{yuan2007model, cai2011constrained}. However, in many biological studies, such as those involving patient or condition-specific data, the number of samples per category is often limited, making robust network estimation challenging~\cite{bertsimas2020certifiably,danaher2014joint}. To address this, we propose an inference framework that leverages a common underlying structure to borrow information across related categories. Specifically, each category-specific network is decomposed into a \textit{global component}, shared across all categories, and a \textit{local component} that captures category-specific variation (see Fig~\ref{fig:cat-schema} for a representative example). By borrowing statistical strength through the global component, our approach improves the estimation accuracy. 

Suppose we observe gene expression data $X_{k,c}$ for each category $c\in\{1,\ldots,C\}$ in population $k\in\{1,\ldots,K\}$. Each matrix $X_{k,c}\in\R^{p\times n_{k,c}}$ contains gene expression measurements for $p$ genes across $n_{k,c}$ samples, with sample covariance matrix $\hat{\Sigma}_{k,c}$. We model the precision matrix $\Theta_{k,c}$ for each population $k$ and category $c$ as the sum of two components: a global component $\Theta^{\mathrm{global}}_{k}$, shared across categories for that particular node in the hypergraph, and a local component $\Theta^{\mathrm{local}}_{k,c}$, specific to category $c$. That is, $$\Theta_{k,c}=\Theta_{k}^{\mathrm{global}}+\Theta^{\mathrm{local}}_{k,c}.$$ 
This decomposition enables the model to capture both spatial structure (through the global components) and category-specific deviations (through the local components).

The optimization problem for categorical inference of the precision matrix is given as
\begin{align*}\label{eq: cat}
    &\hspace{-12mm}\left\{\hat \Theta_k^{global},\hat\Theta_{k,1}^{\mathrm{local}},\ldots,\hat\Theta_{k,C}^{\mathrm{local}}\right\}_{k=1}^K=\\
    &\hspace{-12mm}\argmin_{\left\{ \Theta_k^{\mathrm{global}},\Theta_{k,1}^{\mathrm{local}},\ldots,\hat\Theta_{k,C}^{\mathrm{local}}\right\}_{k=1}^K}
    \underbrace{ 
     \sum_{k=1}^K\sum_{c=1}^C \norm{\Theta_k^{\mathrm{global}}+\Theta_{k,c}^{\mathrm{local}}-\tilde F(\hat\Sigma_{k,c})}^2_{\ell_2} }_{\text{backward mapping deviation}}
     +\underbrace{\gamma \sum_{l>k}W_{kl}\norm{\Theta^{\mathrm{global}}_k-\Theta^{\mathrm{global}}_l}^2_{\ell_2}}_{\text{similarity penalty}}\\
    &\hspace{8mm}+\underbrace{\lambda\sum_{k=1}^K\left[\norm{\Theta_{k}^{global}}_{0,\mathrm{off}}+\sum_{c=1}^C\norm{\Theta^{\mathrm{local}}_{k,c}}_{0,\mathrm{off}}\right]}_{\text{sparsity penalty}}
    +\underbrace{\alpha\sum_{k=1}^K\left[\norm{\Theta_k^{\mathrm{global}}}_{\ell_2}+\sum_{c=1}^C\norm{\Theta_{k,c}^{\mathrm{local}}}_{\ell_2}\right]}_{\text{ridge penalty}}.
    \tag{Categorical ELEM-0}
\end{align*}
Here, the \textit{backward mapping deviation} term ensures that the precision matrix, $\hat\Theta_{k,c}=\hat\Theta^{\mathrm{global}}_{k,c}+\hat\Theta^{\mathrm{local}}_{k,c}$ closely matches the approximate backward mapping $\tilde F^*(\Sigma_{k,c})$ for every $k$ and $c$. The $\ell_0$ penalty enforces sparsity in off-diagonal elements of the global and local components. The \textit{similarity penalty} encourages similarity only among global components according to the structure encoded by $W$. Finally, the \textit{ridge penalty} is used to enforce numerical stability, as required for the optimization algorithm~\cite{bhathena2025parametric}.

The procedure for estimating category-specific precision matrices closely follows the algorithm described previously in Algorithm~\ref{alg}. The main distinction in this setting is the presence of multiple categories, each associated with its local component. The sample covariance matrices ${\hat\Sigma_{k,c}}$ are computed separately for each population $k$ and category $c$, followed by the approximate backward mapping $\tilde F^*(\hat\Sigma_{k,c})$. The dynamic programming routine from~\cite[Algorithm 2]{bhathena2025parametric} is then applied to jointly estimate the global and local components across all index pairs. The lower triangular entries of the precision matrix are then filled in by mirroring the values across the diagonal. Finally, the category-specific precision matrix is formed by summing the corresponding global and local components. The overall procedure has a runtime of $\mathcal{O}\left(CKp^2n+CKp^3+C^2K^2p^2\right)$ and requires $\mathcal{O}\left(CKpn+C^2K^2p^2\right) $ memory. These expressions follow analogously to the proof of Theorem~\ref{thm}, but now account for $C$ distinct categories. 

\subsection{Parameter selection and implementation}

To implement Algorithm~\ref{alg}, we must tune the hyperparameters $\gamma$, $\lambda$, and $\nu_k$, which control spatial smoothness, sparsity, and the soft-thresholding level in the approximate backward mapping, respectively. Let $\hat{\Theta}_k(\gamma, \lambda, \nu_k)$ denote the estimated precision matrix for population $k$ obtained by solving the optimization problem with the set of hyperparameters fixed at $(\gamma, \lambda, \nu_k)$. We define the grid of candidate parameter tuples as
$$
\mathcal{P} = \{(\gamma, \lambda, \nu_k): \gamma \in \Gamma, \lambda \in \Lambda, \nu_k \in N\},
$$
where \(\Gamma, \Lambda\), and \(N\) are discrete sets representing the grid search ranges for the respective parameters.

We select the optimal parameters using the extended Bayesian Information Criterion (eBIC)~\cite{foygel2010extended}, which balances data fit and model complexity:
$$
(\hat \gamma, \hat \lambda, \hat \nu_k) = \argmin_{(\gamma, \lambda, \nu_k) \in \mathcal{P}} \mathrm{eBIC}(\gamma, \lambda, \nu_k),
$$
where
$$
\hspace{-3mm}\mathrm{eBIC}(\gamma, \lambda, \nu_k) = \sum_{k=1}^K n_k\left[\Tr(\hat \Sigma_k \hat \Theta_k(\gamma, \lambda, \nu_k)) - \log \det \hat \Theta_k(\gamma, \lambda, \nu_k)\right] + \log(n_k)\,\mathrm{df}_k + 4\,\mathrm{df}_k \log p, 
$$
and \(\mathrm{df}_k\) denotes the degrees of freedom, defined as the number of nonzero entries in the off-diagonal of \(\hat \Theta_k(\gamma, \lambda, \nu_k)\). The extended BIC has a tunable parameter that can be used to limit the model complexity. In this work, we fix this parameter to a constant value to ensure consistent model selection across experiments. Parameter selection proceeds by performing a grid search over \(\mathcal{P}\), solving the optimization problem for each candidate tuple, and selecting the one that yields the minimum extended BIC. The structure and edge weights of the population hypergraph can be imposed in an ad hoc manner, and are typically learned from data, as we illustrate in our applications. 

For inference across categories, we set the ridge penalty to $\alpha = 0.01$, which enhances the numerical stability of the optimization procedure.

\section{Numerical results on synthetic data}

In this section, we evaluate the performance of our algorithm using a series of synthetic datasets designed to mimic realistic gene network structures. All results on synthetic datasets are averaged over 5 trials. We benchmark our method against JGL~\cite{danaher2014joint}, FASJEM~\cite{wang2017fast}, ELEM-1~\cite{ravikumar2023efficient}, and GRNBoost2~\cite{moerman2019grnboost2}. JGL serves as a classical baseline based on the MLE framework. FASJEM and ELEM-1 are recent, scalable algorithms that have demonstrated strong empirical performance, and GRNBoost2 is a widely used tool in the GRN community. Our approach demonstrates competitive scalability while consistently attaining higher F1 scores, highlighting the advantages of using the $\ell_0$ penalty over its convex surrogates. The Python implementation of our algorithm is publicly available at: 
\url{https://github.com/aareshfb/Network_inference_using_discrete_penalty.git}
.

\subsection{Data generation}\label{sec: data generation}

We adopted a simulation framework similar to~\cite{ravikumar2023efficient} in order to generate synthetic gene networks. Each population is modeled as a disjoint modular network, with $p$ genes divided into $M=10$ modules. Within each module, connectivity follows a power-law degree distribution generated using the Barabási-Albert model~\cite{albert2002statistical}. The resulting modules are combined into a block-diagonal structure, which serves as the true precision matrix for each population. For each network, we simulate data from a zero-mean multivariate Gaussian distribution, varying the sample-to-feature ratio ($n/p$) to assess performance under different data regimes.

To introduce controlled heterogeneity across $K$ populations, we construct a hypergraph using a Minimum Spanning Tree (MST) derived from a $K \times K$ matrix of random positive numbers from the Gaussian distribution. The MST's adjacency matrix encodes the similarity structure among populations. We traverse the MST via breadth-first search, generating each population's covariance matrix by perturbing its parent network: three out of ten modules are modified by adding random weights ( sampled from $\text{Uniform}(-1,1)$) to their nonzero edges. At major branching points (nodes with degree $\geq 3$), we further introduce structural diversity by replacing an unperturbed module with a new power-law module. This process continues until all nodes are visited, defining the true precision matrices for all populations. The topology of this MST determines the degree of similarity among networks. Specifically, the matrix $W$ is constructed as the adjacency matrix of the MST, with entries set to 1 whenever an edge exists between two nodes in the MST, and 0 otherwise. Our experiments systematically vary both the $n/p$ ratio and the size of the hypergraph to assess algorithmic scalability and accuracy.

\subsection{Performance metrics}

We assess network reconstruction accuracy using standard metrics: precision, recall, and F1-score. Precision measures the proportion of predicted edges that are correct. It is calculated by $\precision = \frac{\TP}{\TP + \FP}$, where $\TP$ denotes true positives (correctly inferred edges) and $\FP$ denotes false positives (incorrectly inferred edges). The recall quantifies the proportion of true edges that are recovered. It is given as $\recall = \frac{\TP}{\TP + \FN}$, with $\FN$ representing false negatives (true edges that were missed). The F1 score, defined as $\mathrm{F1} = 2 \frac{\precision \cdot \recall}{\precision + \recall}$, provides a balanced summary of both metrics.

The F1 score is particularly informative for sparse graphical models, where both false positives (over-prediction) and false negatives (missed edges) can significantly impact downstream biological interpretation. By capturing the trade-off between precision and recall, the F1 score serves as a comprehensive measure of inference quality in our experiments.

\subsection{Simulation results}

In our first experiment, we compare the proposed method, ELEM-0, with ELEM-1, FASJEM, JGL, and GRNBoost2. We use the aforementioned extended BIC-based approach for parameter selection for our proposed ELEM-0, as well as for ELEM-1, FASJEM, and JGL. GRNBoost2 does not have any tunable hyperparameters.
The results are presented in Table~\ref{tab:grn_metrics}. In this setting, we fix the number of genes at $p = 250$, the sample size per population at $n/p = 20$, and the total number of populations at $K = 10$.

\begin{table}[h]
\centering
\begin{tabular}{l|ccccc}
\hline
\textbf{Metric} & \textbf{ELEM-0} & \textbf{ELEM-1} & \textbf{FASJEM} & \textbf{JGL} & \textbf{GRNBoost2} \\
\hline
F1-score      &0.88           &0.77         &0.23             &0.26          &0.00          \\
Precision     &0.88           &0.80         &0.14             &0.43          &0.00          \\
Recall        &0.88           &0.74         &0.54             &0.19          &1             \\
Time (s)      &32.66          &6.90         &26.61            &2345.17       &401.53        \\
\hline
\end{tabular}
\caption{\textbf{Comparison of metrics across different GRN inference methods.} Results are averaged over five trials. In these experiments, the number of populations is set to $K = 10$, the number of genes to $p = 250$, and the sample size per population to $n/p = 20$.
}
\label{tab:grn_metrics}
\end{table}

ELEM-0 achieves an F1 score of 0.88, along with high precision and recall, outperforming all other methods. In contrast, FASJEM, JGL, and GRNBoost2 struggle to accurately recover the true network structure, yielding substantially lower F1 scores.
The fourth row of Table~\ref{tab:grn_metrics} reports the average runtime (in seconds) for each algorithm. ELEM-0 achieves a runtime of 32.66 seconds, which is comparable to ELEM-1 and FASJEM. The slightly slower runtime of ELEM-0 relative to ELEM-1 and FASJEM is expected, as both methods rely on convex formulations based on the $\ell_1$ penalty, resulting in more tractable optimization problems. In contrast, ELEM-0 directly addresses the nonconvexity introduced by the $\ell_0$ penalty, leading to slightly higher computational cost but improved performance.  By comparison, JGL and GRNBoost2 exhibit substantially larger runtimes.

Our next set of experiments, summarized in Fig~\ref{fig:exp1}, investigates the effect of the sample-to-feature ratio ($n/p$) and the number of populations ($K$) on the error metrics. In the plots, solid lines represent the mean F1 scores, precision, and recall, while the shaded regions indicate the range between the minimum and maximum values, capturing variability across runs. FASJEM, JGL, and GRNBoost2 consistently yield F1 scores below 0.3 across all tested conditions and are therefore omitted from the plots for clarity. 

The first row of Fig~\ref{fig:exp1} illustrates the performance of ELEM-0 and ELEM-1 across varying $n/p$ ratios. In this experiment, we set $p = 2000$ to simulate a high-dimensional setting representative of real-world gene expression data. The number of populations is fixed at $K = 20$, and $n/p$ is varied from $0.5$ to $30$. For low $n/p$ ratios (i.e., fewer samples per feature), both methods show reduced F1-scores. However, ELEM-0 significantly outperforms ELEM-1 across all trials. When $n/p \geq 5$, ELEM-0 consistently achieves an F1-score around 0.88, while ELEM-1 peaks at a lower F1-score of approximately 0.80. The second row of Fig~\ref{fig:exp1} shows how performance varies with the number of populations $K$. In this experiment, we fix $n/p = 20$, set $p = 250$, and vary the number of populations $K$ from $3$ to $100$. ELEM-0 maintains a stable and high F1-score, ranging between $0.88$ and $0.91$, demonstrating robustness to increasing $K$. In contrast, ELEM-1 exhibits a lower F1-score than ELEM-0 for all population sizes.
 
\begin{figure}
    \hspace{-7mm}\includegraphics[width=0.36\linewidth]{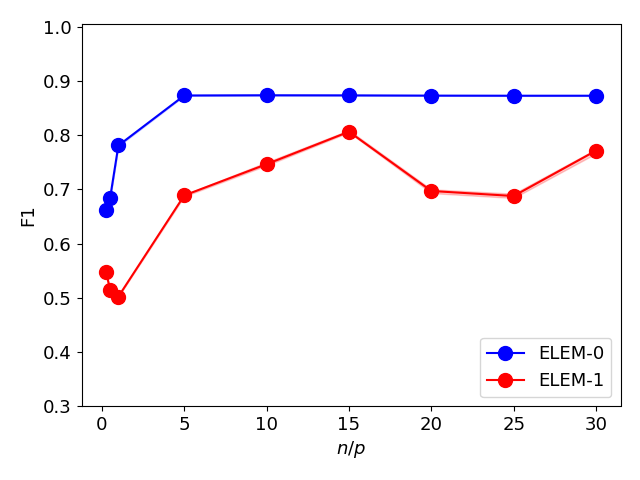}\includegraphics[width=0.36\linewidth]{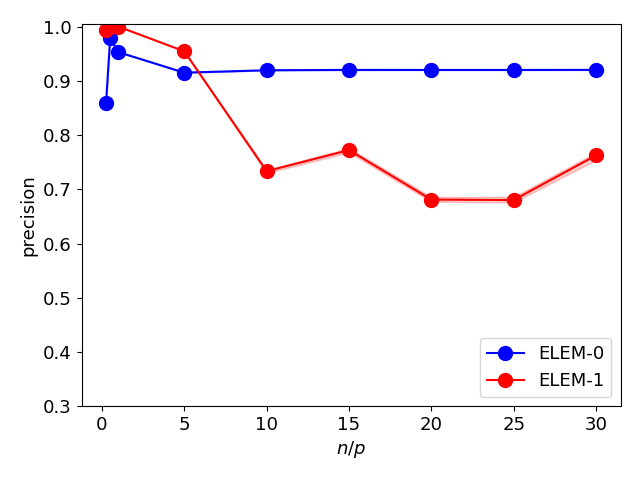}\includegraphics[width=0.36\linewidth]{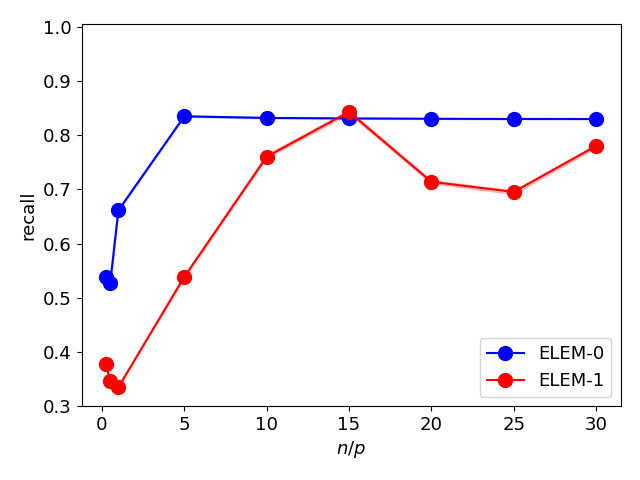}\\
    
    \hspace{-7mm}\includegraphics[width=0.36\linewidth]{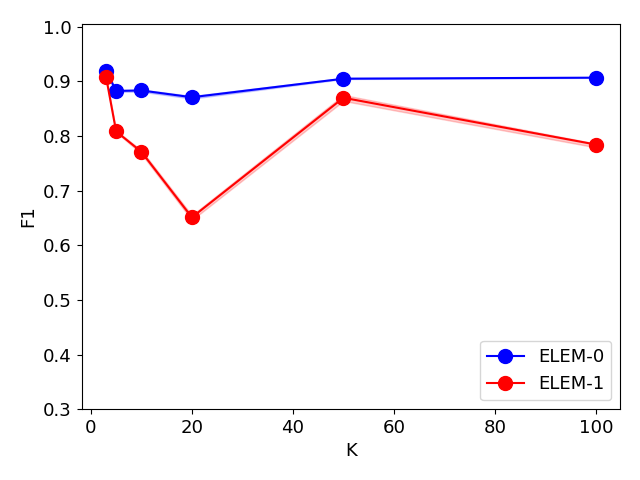}\includegraphics[width=0.36\linewidth]{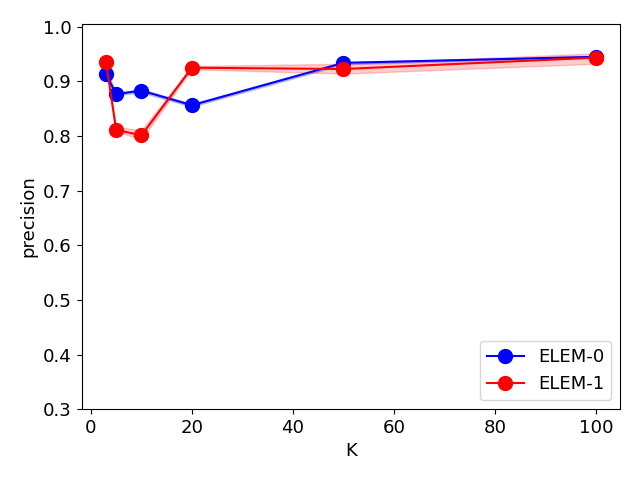}\includegraphics[width=0.36\linewidth]{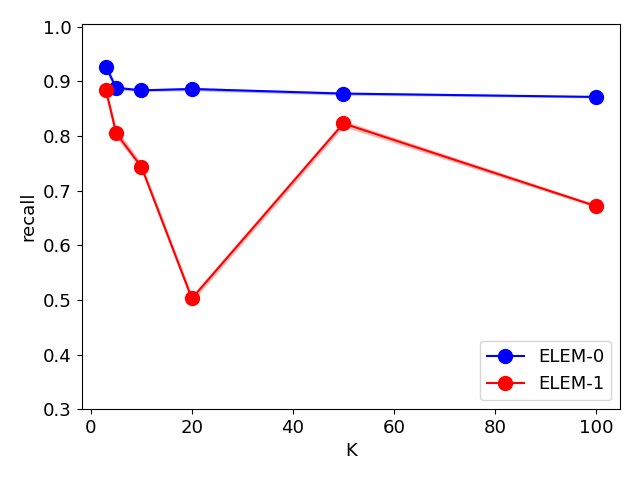}
    \caption{\textbf{Comparative performance of the proposed ELEM-0 (with $\ell_0$-penalty) and ELEM-1 on synthetic gene network benchmarks.} The first row shows performance metrics for varying $n/p$ ratios, while the second row presents performance metrics for different numbers of populations $K$. Solid lines represent the mean F1 scores, precision, and recall, while the shaded regions indicate the range between the minimum and maximum values
    }
    \label{fig:exp1}
\end{figure}

\subsection{Extension to inference over multiple categories} 

In this section, we present experiments on categorical inference, where the goal is to recover both shared and category-specific components of the underlying networks. We evaluate the performance of the categorical ELEM-0 algorithm (introduced in Section~\ref{sec: categorical ELEM-0}) and compare it with the standard ELEM-0 algorithm. We begin by generating a global network for each hypergraph node using the procedure described in Section~\ref{sec: data generation}. To introduce category-specific variation, we perturb a fixed proportion of the off-diagonal entries in the global precision matrix by adding randomly weighted edges. This setup allows us to control the degree of similarity between categories and evaluate the robustness of our method under structured heterogeneity. It reflects biological scenarios such as cell-type–specific effects of a disease or niche-specific variations between primary and recurrent disease. By introducing controlled inter-category differences, we simulate the heterogeneity commonly observed across patient cohorts within spatial or latent regions in real-world transcriptomic data.

To quantify the extent of variation across categories, we introduce the \textit{local edge ratio} $\delta$, which specifies the proportion of local (i.e., category-specific) edges relative to the total number of edges in each network. A small value of $\delta$ indicates that most edges are shared across categories, and only a few are specific to individual networks. In contrast, a large $\delta$ implies that each category-specific network contains a substantial number of local edges, making it more distinct from the shared global structure. For example, $\delta=50\%$ means that half of the edges in each network are shared, and half are local.

In the first experiment, illustrated in Fig~\ref{fig:F1_vs_pertub_strength}, we study the effect of the local edge ratio $\delta$ on the error metrics of the inferred precision matrices. We focus on recovering the networks for two categories. In the plot, dots and triangles denote the error metrics for the individual categories, while the solid line represents their mean. For this experiment, we set $p=250$, $K=5 $ and $n/p=20$. 

\begin{figure}
        \hspace{-9mm}\includegraphics[width=0.37\linewidth]{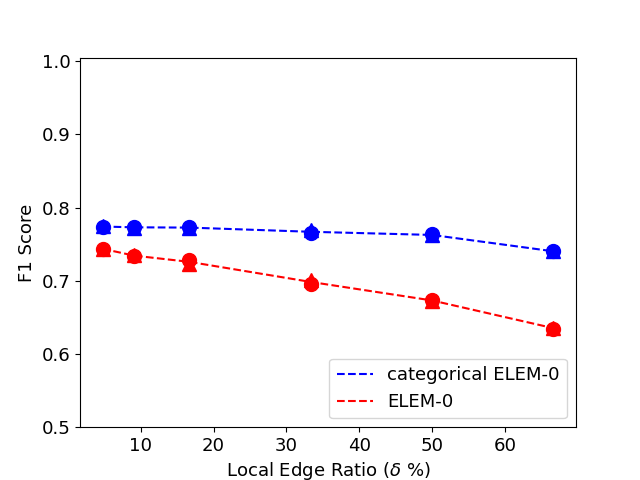}\includegraphics[width=0.37\linewidth]{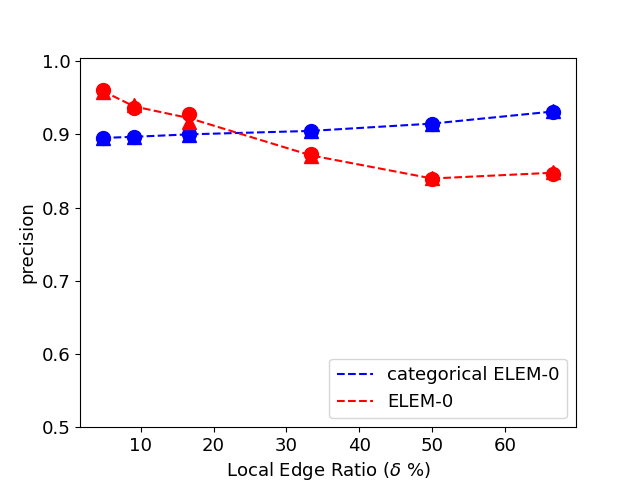}\includegraphics[width=0.37\linewidth]{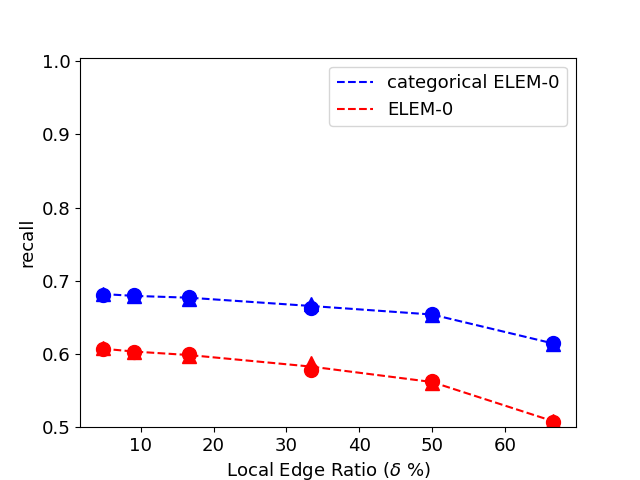}
        \caption{\textbf{Comparison of categorical and standard ELEM-0 across increasing local edge ratios} Performance metrics are shown for both categorical ELEM-0 and standard ELEM-0 across increasing values of $\delta$. Dots and triangles denote the error metrics for the individual categories, while the solid line represents their average. }
        \label{fig:F1_vs_pertub_strength}
\end{figure}

For a small local edge ratio (low $\delta$), standard ELEM-0 performs comparably to categorical ELEM-0, likely because the networks are highly similar and the benefit of sharing information is minimal. However, as dissimilarity increases (higher $\delta$), categorical ELEM-0 consistently outperforms standard ELEM-0. This improvement arises because the categorical algorithm effectively pools information for the shared network structure, providing more robust estimates of common edges even as category-specific differences grow. When the number of local edges increases, independent inference must estimate more parameters from limited data, which can lead to higher variance and reduced accuracy. In contrast, categorical ELEM-0 leverages the shared structure as a statistical regularizer, stabilizing network recovery and mitigating the impact of noise or limited sample size. Moreover, since the local edges are introduced independently of the hypergraph topology, the signal distinguishing categories is not aligned with the population-level structure, making it harder for independent inference to distinguish true unique edges from noise. By explicitly modeling both shared and local components, categorical ELEM-0 disentangles these effects, yielding higher F1 scores as local edge ratio increases. 

We next evaluate the performance of the categorical ELEM-0 across a range of $n/p$ ratios and $K$, using a fixed local edge ratio of $\delta = 50\%$ and $p=250$. Once again, we compare categorical ELEM-0 with standard ELEM-0. Results are presented in Fig~\ref{fig:experiment shared inference}. The top row shows the error metrics across different values of $n/p$. In this setting, we fix the number of populations at $K = 20$. For $n/p \geq 15$, the standard ELEM-0 method consistently recovers networks with an F1 score around 0.7. In contrast, the categorical ELEM-0 initially exhibits lower F1 scores for small $n/p$ values, but its performance steadily improves as $n/p$ increases. Notably, for $n/p \ge 20$, categorical ELEM-0 outperforms the standard approach, reaching an F1 score of approximately $0.8$ when $n/p = 30$.

The bottom row of Fig~\ref{fig:experiment shared inference} reports F1 scores across varying values of $K$ for fixed $n/p=20$. The standard ELEM-0 method achieves an F1 score of around 0.7 for small $K\le 20$, but its performance degrades as $K$ increases. When $K=100$, standard ELEM-0 recovers networks with F1-scores below $0.4$. In contrast, categorical ELEM-0 maintains consistently higher performance, with F1 scores ranging from approximately 0.8 to 0.75. This demonstrates the improved robustness of categorical ELEM-0 for a larger number of populations.

\begin{figure}
    \hspace{-5mm}\includegraphics[width=0.36\linewidth]{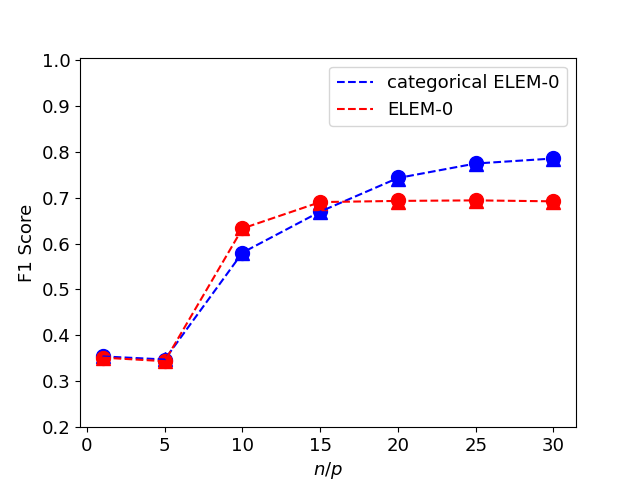}\includegraphics[width=0.36\linewidth]{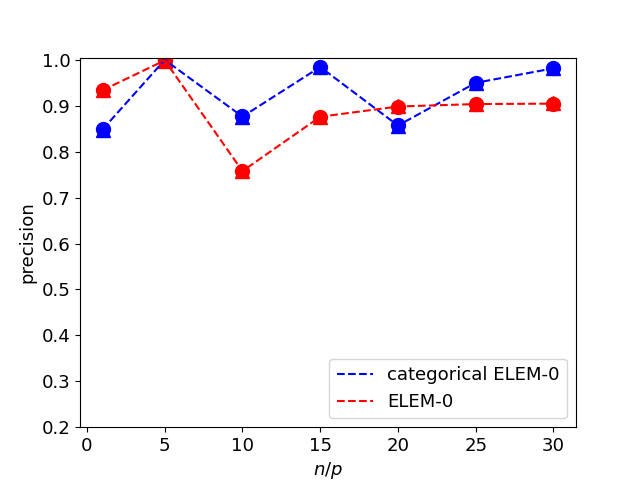}\includegraphics[width=0.36\linewidth]{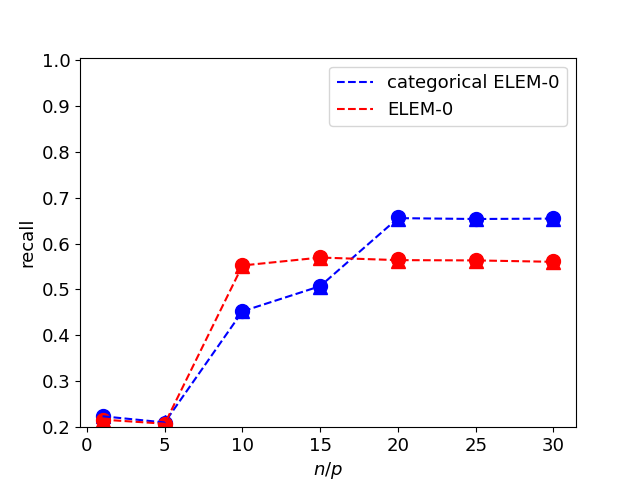}\\
    
    \hspace{-5mm}\includegraphics[width=0.36\linewidth]{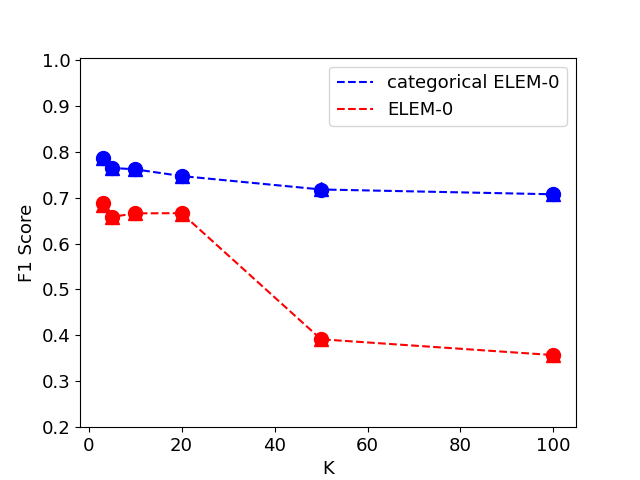}\includegraphics[width=0.36\linewidth]{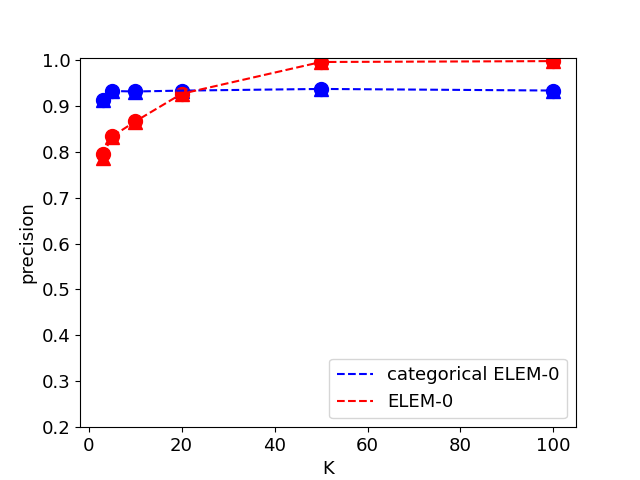}\includegraphics[width=0.36\linewidth]{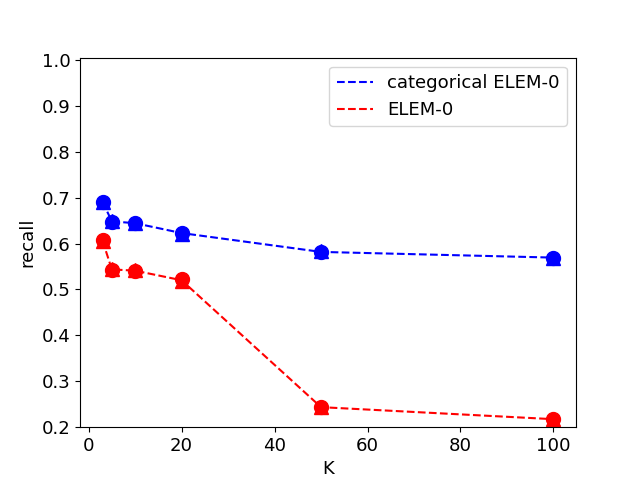}
    \caption{\textbf{Performance comparison of categorical ELEM-0 and standard ELEM-0 at $\boldsymbol{\delta = 50\%}$.}  The first row shows performance metrics for varying $n/p$ ratios, while the second row presents performance metrics for different values of $K$.  Dots and triangles denote the error metrics for the individual categories, while the solid line represents their average. }
    \label{fig:experiment shared inference}
\end{figure}

These results highlight the robustness and versatility of our categorical inference algorithm across diverse data regimes and population structures, underscoring its potential as a powerful tool for investigating complex disease biology and context-specific network rewiring.

\section{Joint inference of gene regulatory networks in glioblastoma tumor populations}

\begin{figure}
    \centering
     \includegraphics[height=7in]{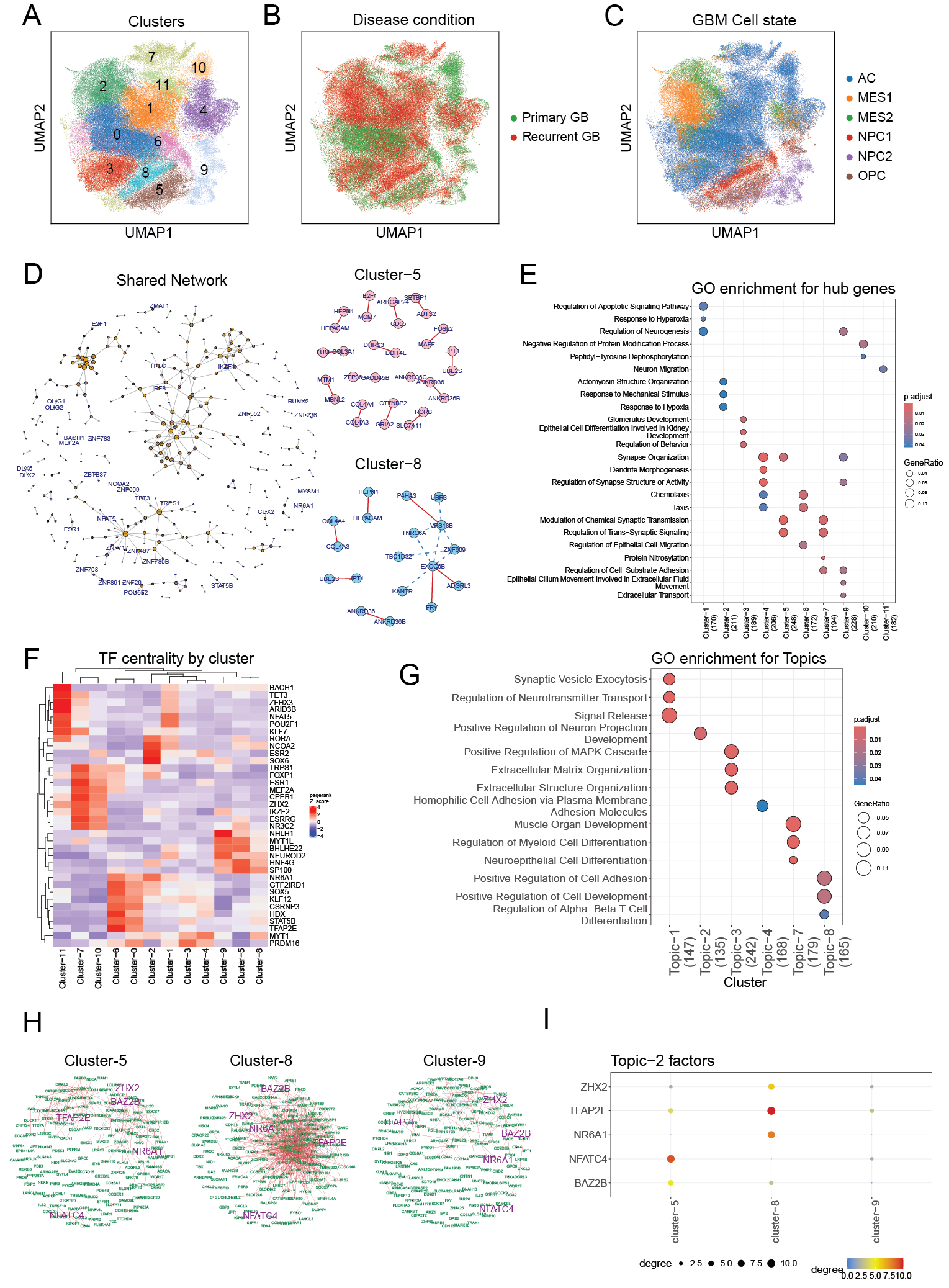}
    \caption{\footnotesize\textbf{Joint network inference reveals shared and population-specific gene regulatory architecture across glioblastoma tumor populations.}(A) UMAP embedding of cancer cell clusters across 24 patients (B) Distribution of primary and recurrent tumor cells across clusters (C) Neftel states of tumor cells (D) Inferred gene regulatory networks: the global network backbone (edges shared across all clusters) is shown alongside population-specific regulatory edges unique to clusters 5 (OPC-like) and 8 (NPC1-like). Red solid lines indicate positive regulatory interactions; blue dashed lines indicate inhibitory interactions. (E) Gene Ontology enrichment analysis for the top 5\% hub genes (composite centrality ranking) from each cluster-specific network, revealing functional specialization of regulatory programs. (F) Top transcription factors (TFs) in each network ranked by Pagerank centrality, highlighting distinct TF combinations central to different tumor subpopulations. (G) Topic modeling (LDA) applied to progenitor-cell clusters identifies eight gene modules (topics) associated with distinct biological functions. (H) The regulatory network for Topic 2, uniquely active in NPC1-like cells, is shown between progenitor clusters. (H) Key TFs in the Topic 2 network ranked by node degree, pinpointing central regulators of neural precursor identity. }
    \label{fig:gbm-sc}
\end{figure}

\begin{figure}
    \centering
    \includegraphics[height=4.5in]{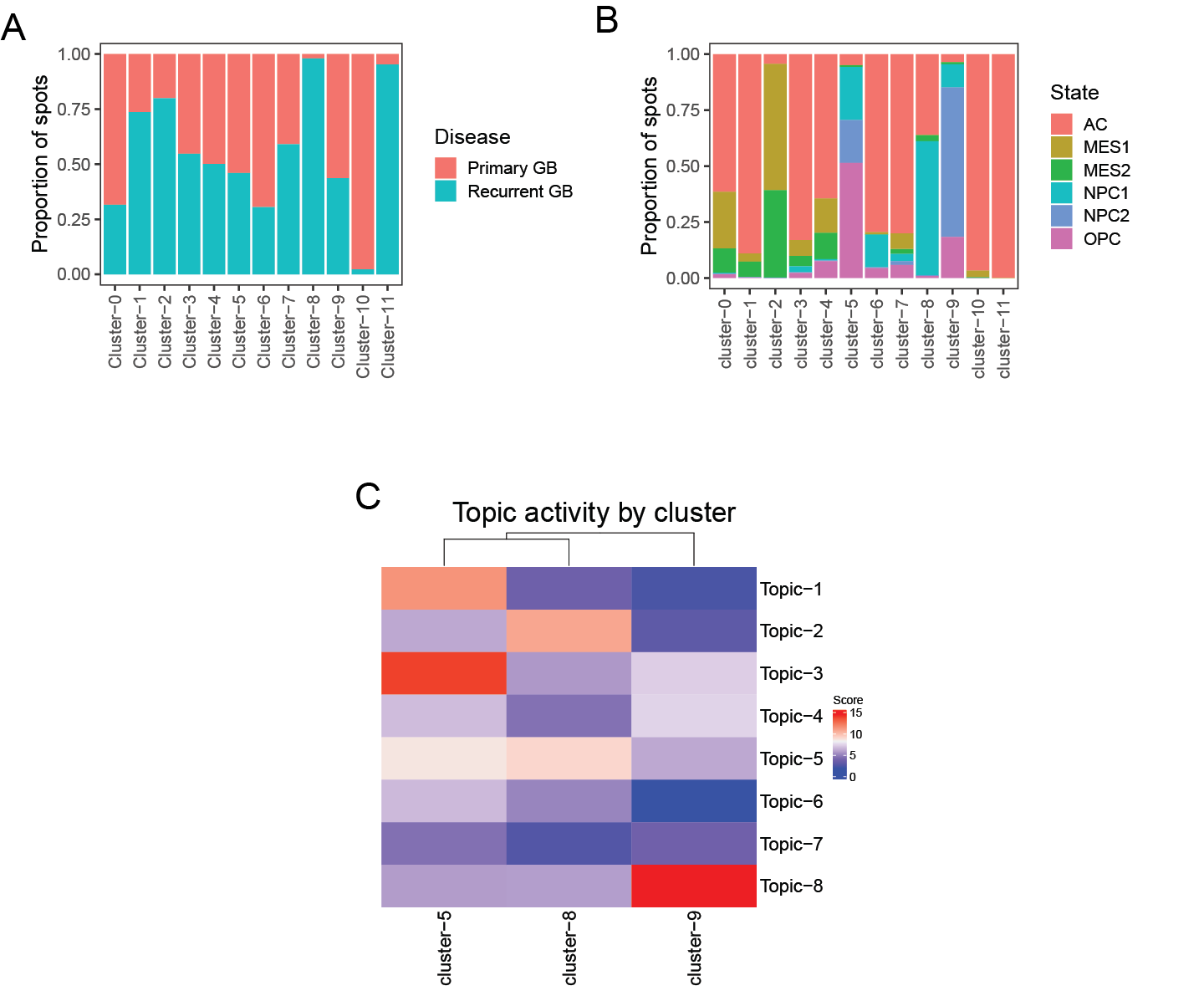}
    \caption{\footnotesize\textbf{Supplemenatary to Figure~\ref{fig:gbm-sc}:} (A) relative distribution of primary and recurrent tumor cells in each Leiden cluster. (B) Distribution of Neftel states by cluster. (C) Enrichment for different TF topic modules in precursor-type clusters. }
    \label{fig:fig1suppl}
\end{figure}

Glioblastoma is a highly aggressive malignancy characterized by pronounced inter- and intra-patient heterogeneity \cite{eisenbarth2023glioblastoma}. While numerous subtyping strategies have been developed---ranging from clustering of gene expression programs \cite{neftel2019integrative} to pathway activity profiling \cite{garofano2021pathway}---recent evidence indicates that glioblastoma cell states are primarily shaped by dynamic rewiring of gene regulatory networks in response to microenvironmental cues, rather than by static genomic alterations \cite{yabo2022cancer}. To systematically interrogate this network plasticity and identify key transcriptional regulators, we applied our network inference algorithm to a comprehensive, previously unpublished single-cell RNA-seq atlas of glioblastoma from the University of Michigan Medical School. The dataset comprises single-cell transcriptomes from 24 patients (13 primary, 11 recurrent tumors), yielding approximately 450,000 high-quality cells. Batch effects were corrected using scVI, and broad cell types were annotated based on canonical marker gene expression. For downstream analysis, we focused exclusively on tumor cell populations, excluding immune and neural cells, resulting in a dataset of 148,019 tumor cells (66,671 primary; 81,348 recurrent). To capture the full extent of tumor heterogeneity, we performed unsupervised clustering using scVI latent embeddings to construct a neighborhood graph, followed by Leiden clustering (resolution = 0.5). After removing clusters with fewer than 150 cells, we identified 12 distinct tumor cell clusters distributed across the gene expression latent space (Fig~\ref{fig:gbm-sc}A). Most clusters exhibited a balanced mix of primary and recurrent disease (Fig \ref{fig:fig1suppl}A), with cluster 10 enriched for primary tumors and clusters 8 and 11 enriched for recurrent disease (Fig~\ref{fig:gbm-sc}B, \ref{fig:fig1suppl}A). To facilitate biological interpretation, we scored cells for Neftel tumor state marker genes and assigned clusters to subtypes accordingly (Fig~\ref{fig:gbm-sc}C, \ref{fig:fig1suppl}B). The majority of tumor cells were classified as AC-like, with MES-like and precursor-like states occupying distinct regions of the embedding. Clusters 5, 8, and 9 were enriched for precursor states, while cluster 2 was predominantly MES-like (Fig~\ref{fig:gbm-sc}C, \ref{fig:fig1suppl}B).

To infer gene regulatory networks, we leveraged the relative proximity of clusters in the scVI latent space to impose similarity constraints, enabling joint inference across related populations. We first identified a global network backbone by extracting edges shared across all clusters (Fig~\ref{fig:gbm-sc}D). This core network comprised 406 nodes and 446 edges, with prominent interactions between known functional partners such as \textit{OLIG1/OLIG2}, \textit{DLX2/DLX5}, and \textit{TUBB2A/TUBB2B}. Larger connected components involving immune response and cell cycle genes were also observed, reflecting conserved regulatory programs across tumor populations. To highlight cluster-specific regulatory rewiring, we visualized the top 15 unique edges in two precursor-enriched clusters (5 and 8). Cluster 5, enriched for OPC-like cells, and cluster 8, enriched for NPC1-like cells, both exhibited strong regulatory interactions involving collagens and cell adhesion molecules (e.g., \textit{HEPACAM}), consistent with roles in axonal and neurite outgrowth \cite{wareham2024collagen}. Notably, ankyrin-repeat domain-containing gene family members, which scaffold multiprotein complexes for cytoskeletal organization, were prominent in both clusters \cite{kurochkina2018phosphorylation}. Cluster 5 displayed regulatory interactions such as \textit{E2F1-MCM7} (cell cycle), \textit{AUTS2-SETBP1}, and \textit{MBNL2-MTM1} (brain development), while cluster 8 featured numerous links among genes involved in intracellular signaling and vascular trafficking, critical for neurodevelopment \cite{wen2013phenotypic, vacca2024exploring}.

To systematically prioritize key regulators, we adopted the approach of Iacono et al. \cite{iacono2019single}, identifying the top 5\% of hub genes in each network using five centrality measures: degree, betweenness, closeness, eigenvector, and pagerank centrality. This composite ranking yielded robust lists of population-specific hub genes, which were subjected to Gene Ontology enrichment analysis (Fig~\ref{fig:gbm-sc}E). MES-like clusters were enriched for hubs involved in hypoxia response, oxygen stress, and ECM organization, while precursor-enriched clusters showed enrichment for synaptic signaling, vascular transport, and neurogenesis. Cluster 11, specific to recurrent tumors, was enriched for hubs regulating neuronal migration, potentially contributing to tumor invasiveness. Other clusters were characterized by hubs involved in neuronal development, function, and cell migration. Transcription factors (TFs) were further ranked by pagerank centrality within each network, and the most variable TFs across clusters were visualized (Fig~\ref{fig:gbm-sc}F). Distinct TF combinations were central to different clusters, consistent with hub gene enrichment results. Notably, cluster 11, which was enriched for neuronal migration processes and recurrent disease, exhibited high centrality for \textit{BACH1}, a TF implicated in glioblastoma treatment resistance and invasiveness by sustaining stem-cell identity \cite{hu2024transcription, wang2024single}. Progenitor-like clusters were driven by TFs regulating neuronal differentiation (\textit{NEUROD2, NHLH1, MYT1L, BHLHE22}) and chromatin remodeling (\textit{SP100}), while MES-like clusters were characterized by hypoxia response TFs (\textit{RORA, ESR2, NCOA2}).

To further dissect regulatory logic among precursor cell populations, we applied Latent Dirichlet Allocation (LDA) topic modeling to the inferred networks, following previous frameworks for gene network analysis \cite{lou2020topicnet, zhang2023inference}. In this approach, each TF's regulome is treated as a document and genes as words, enabling the identification of latent topics corresponding to functional gene modules. The aggregated regulomes of all TFs were used to fit an LDA model (using the \textit{topicmodels} package in R \cite{grun2011topicmodels}), with the optimal number of topics ($k=8$) determined via the \textit{ldatuning} package \cite{nikita2016package}. Genes were assigned to topics based on relative weights, and the sum of edge connectivity within each topic module was used as a proxy for module activity. Enrichment analysis revealed that the eight topics mapped onto distinct biological processes, including ECM organization, neurotransmitter transport, vesicle exocytosis, differentiation, and cell adhesion (Fig~\ref{fig:gbm-sc}G). Notably, Topic 2 was uniquely active in cluster 8, which harbors the NPC-I like cells (Fig \ref{fig:fig1suppl}C). Visualization of the Topic~2 subnetwork across relevant clusters (Fig~\ref{fig:gbm-sc}H) revealed extensive connectivity in Cluster~8, highlighting the centrality of this module in neural precursor cell identity. To identify key regulators of this module, we ranked TFs by their connectivity within the Topic~2 subnetwork. Fig~\ref{fig:gbm-sc}I presents the top TFs by degree, demonstrating that factors such as \textit{ZHX2}, \textit{TFAP2E}, and \textit{NR6A1} are highly connected in Cluster~8. These TFs are known to maintain neural stemness and regulate cell fate specification in precursor cells \cite{wu2009zhx2, hong2014transcription, li2024expanding}, supporting their putative roles as drivers of the observed regulatory program.

In summary, our integrative framework enables efficient, simultaneous inference of large-scale gene regulatory networks across multiple tumor cell populations, capturing both shared and context-specific regulatory architecture. Through a combination of network centrality analysis, functional enrichment, and topic modeling, we systematically identify key transcription factors and gene modules underlying glioblastoma cell state diversity. This approach not only recapitulates known biology but also uncovers novel regulators and interactions that define the molecular identity of distinct tumor subpopulations, providing a foundation for mechanistic studies and therapeutic targeting in glioblastoma.\vspace{2mm}

\section{Gene network rewiring along continuous spatial gradients}

\begin{figure}
    \centering
    \includegraphics[height=6.9in]{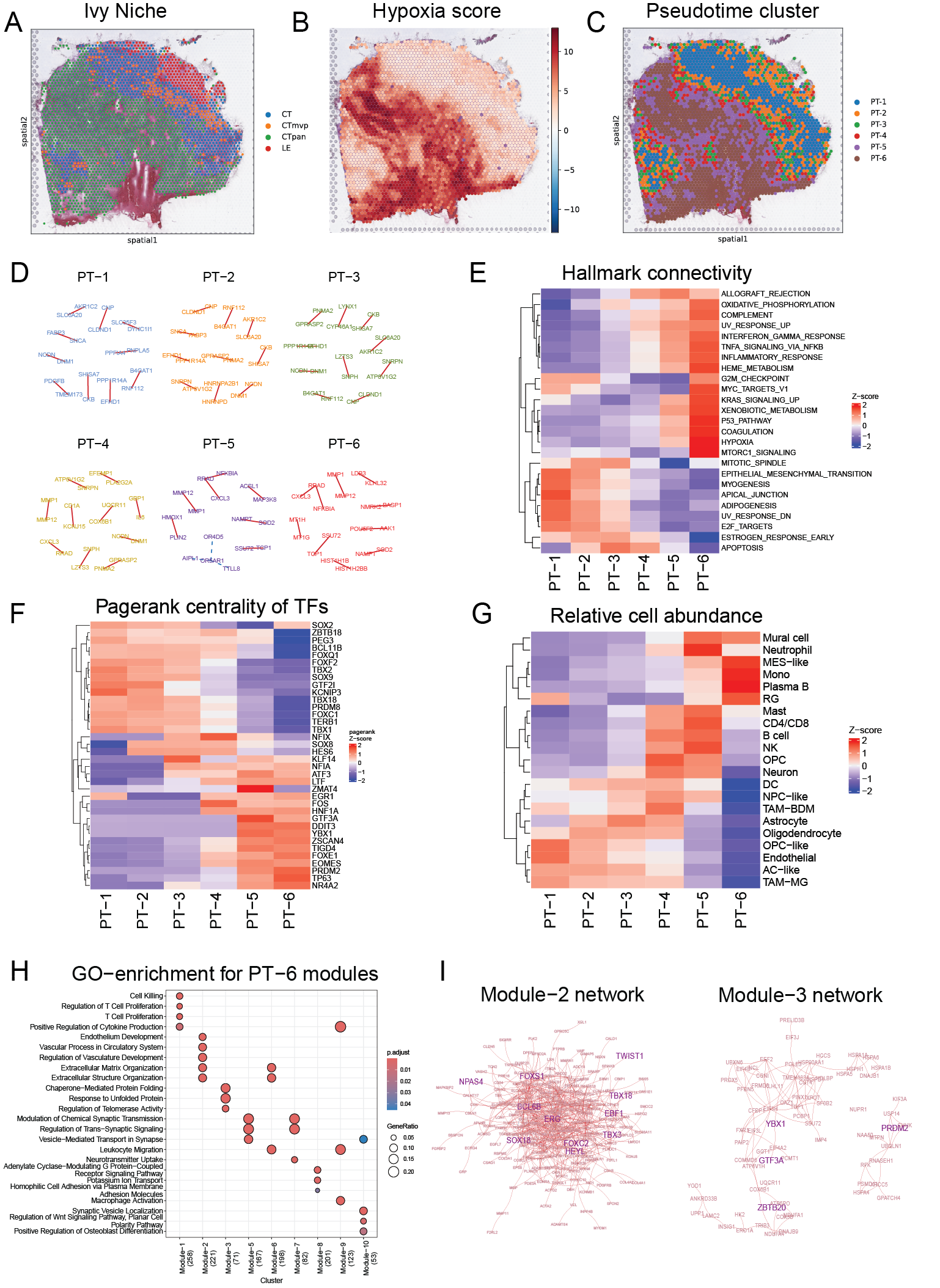}
    \caption{\footnotesize\textbf{Gene regulatory network rewiring along continuous hypoxia gradients in glioblastoma spatial transcriptomics.} (A) Annotation of anatomical niches on a representative Visium tissue slide, delineating perivascular, tumor core, and perinecrotic regions. (B) Spatial distribution of hypoxia response activity, quantified by enrichment scores for the Hallmark Hypoxia gene set in each spot. (C) Stratification of tissue spots into six populations (PT1–PT6) along an increasing hypoxic stress gradient, from perivascular to hypoxic/necrotic regions. (D) Top 10 strongest regulatory edges in each population-specific (PT) network, highlighting progressive rewiring of gene interactions across the hypoxia gradient. (E) Scaled activity scores for MSigDB Hallmark gene sets in each PT network, revealing dynamic shifts in pathway connectivity with increasing hypoxia. (F) Top transcription factors (TFs) in each network ranked by Pagerank centrality, illustrating changes in TF importance and regulatory dominance along the gradient. (G) Estimated spot composition for each PT cluster, based on cell type deconvolution using a single-cell reference atlas. (H) Gene Ontology enrichment for modules identified in the PT6 (most hypoxic) network, demonstrating functional specialization of gene communities. (I) Connectivity of Modules 2 and 3 within the PT6 network, highlighting regulatory interactions related to hypoxia-driven ECM remodeling and vascular development, and stress response activities respectively. }
    \label{fig:hypxoia}
\end{figure}

Hypoxia is a molecular hallmark of glioblastoma, critically contributing to treatment resistance, tumor invasion, and the establishment of an immunosuppressive microenvironment \cite{park2022current}. Recent spatial transcriptomics studies have demonstrated that hypoxic stress not only organizes the global tumor architecture in glioblastoma but also shapes the spatial distribution of both immune and malignant cell states \cite{greenwald2024integrative}. Notably, distinct hypoxia-adapted tumor states---such as MES-Hypoxia, MES-Astrocytic, and Chromatin Regulatory---have been described, representing adaptive rewiring of canonical malignant programs in response to microenvironmental gradients. Elucidating the molecular drivers of these dynamic regulatory changes is essential for identifying vulnerabilities in treatment-resistant tumor populations. To systematically investigate gene regulatory network (GRN) dynamics along a hypoxic gradient, we performed Visium spatial transcriptomics on 16 slides spanning primary and recurrent glioblastoma. After stringent quality control, we retained high-quality spots for downstream analysis. To integrate data across multiple slides and mitigate batch effects, we trained a scVI model and obtained batch-corrected, imputed gene expression counts. Hypoxia signaling activity was quantified in each spot using enrichment scores for the Hallmark Hypoxia signature. Based on these scores, spots were stratified into six populations (PT1–PT6) along the hypoxia gradient, representing increasing hypoxic stress. Spatial visualization revealed that PT1–PT3 were localized near vasculature, with a progressive transition toward the tumor core and perinecrotic regions in PT4–PT6 (Fig\ref{fig:hypxoia}A–C).

We applied our joint network inference algorithm to gene expression data from each hypoxia-defined population, enabling simultaneous estimation of GRNs while leveraging population similarity. Analysis of the top 10 strongest edges in each network revealed persistent interactions in PT1–PT3 involving cell cycle (\textit{CCND1}), metabolic (\textit{FABP3, CKB, AKR1C2}), and neuronal signaling genes (\textit{NCDN, EFHD1}), indicating conserved biological activity in perivascular regions. Notably, \textit{PDGFB} activity and ribonucleoprotein complex interactions were prominent near vasculature. PT4 marked a transition, with emergent connectivity among matrix remodeling and immune signaling genes. By PT5, the strongest edges were unique to this population, featuring activation of stress response genes (\textit{SOD2, HMOX1, TCP1}). In PT6, novel interactions among histone genes suggested chromatin remodeling and epigenetic adaptation to severe hypoxia (Fig~\ref{fig:hypxoia}D). To assess the functional impact of network rewiring, we quantified pathway activity as the sum of edge weights among genes within curated Hallmark gene sets~\cite{liberzon2015molecular} for each population (Fig~\ref{fig:hypxoia}E). This network-centric approach provides a nuanced measure of pathway activation, sensitive to context-specific regulatory interactions \cite{liang2023pathway}. Early populations (PT1–PT3) exhibited strong connectivity within EMT and vascular programs, consistent with perivascular stemness and invasiveness \cite{iwadate2016epithelial, uribe2022adapt}. PT4 represented an intermediate state, with attenuation of vascular programs and activation of hypoxia-associated pathways. PT5–PT6 displayed pronounced activation of hypoxia, p53, and mTORC1 signaling, as well as metabolic and immune modules \cite{monteiro2017role}. Pagerank centrality analysis of transcription factors (TFs) revealed a sharp regulatory shift along the gradient (Figure~\ref{fig:hypxoia}F): perivascular populations were dominated by stemness and developmental TFs (\textit{SOX2, SOX9, PRDM8, T-box family}), while hypoxic populations were enriched for stress response (\textit{ATF3, YBX1, DDIT3, NR4A2, TP63}) and EMT-associated FOX-family factors \cite{katoh2013cancer}.

Given the compositional heterogeneity inherent to spatial transcriptomics spots, we performed cell type deconvolution using a reference single-cell atlas \cite{ruiz2022harmonized} to estimate the abundance of tumor and immune cell states within each spot. This analysis revealed pronounced compositional shifts along the hypoxia gradient (Fig~\ref{fig:hypxoia}G): perivascular regions were enriched for AC- and OPC-like tumor states, while MES-like states predominated in hypoxic, necrotic regions, consistent with recent findings \cite{greenwald2024integrative, haley2024hypoxia, neftel2019integrative}. The enrichment of immune signaling pathways in later populations may partially reflect altered spot composition. To further resolve the functional architecture of hypoxic networks, we applied the Leiden community detection algorithm to the PT6 network. Leiden clustering partitions large-scale gene regulatory networks into densely connected modules, enabling the identification of functionally coherent gene communities even in the presence of spot-level heterogeneity. Importantly, because spatial transcriptomic spots often contain mixtures of tumor and non-tumor cells, focusing on network connectivity rather than expression alone allows Leiden clustering to reveal tumor-intrinsic regulatory modules that may otherwise be masked by compositional variation. We identified ten modules (minimum 50 genes each), which were subjected to Gene Ontology enrichment analysis. Modules were specialized for a range of biological processes including migration and chemotaxis, immune signaling, neuronal function, ECM organization and angiogenesis, and cellular stress response (Fig~\ref{fig:hypxoia}H). Module-2 was characterized by extensive connectivity among transcription factors implicated in hypoxia-driven ECM remodeling and angiogenesis, while Module-3 was centered on \textit{YBX1}, a critical regulator of survival and radioresistance in glioma stem cells \cite{zheng2023multiomics} (Fig~\ref{fig:hypxoia}I). This modular decomposition not only enhances the interpretability of large-scale regulatory networks but also facilitates the identification of candidate master regulators within specific biological contexts, providing a focused framework for mechanistic investigation and potential therapeutic targeting.

Our approach enables high-resolution mapping of gene network rewiring along physiological gradients, revealing both conserved and niche-specific regulatory modules. By integrating community detection and centrality analyses, we identify candidate transcriptional regulators and gene modules that may underlie cell fate transitions, therapy resistance, and spatial adaptation in glioblastoma. This framework is readily extensible to single-cell resolution spatial transcriptomics datasets, offering a powerful tool for dissecting the molecular logic of tumor adaptation to varying microenvironments.\vspace{2mm}

\section{Niche-Specific Gene Network Rewiring Between Primary and Recurrent Disease in Glioblastoma}

To investigate the regulatory mechanisms driving spatially localized adaptation to therapy and drivers of recurrence in glioblastoma, we developed and benchmarked a categorical extension of our joint network inference algorithm that let us identify niche-specific differential interactions in recurrent tumor slides as shown in Figure~\ref{fig:cat-schema}. ST slides were annotated using the IVY-GAP marker genes, and regional spots were aggregated across all primary and recurrent samples for inference. To address batch effects and technical dropout, we applied scVI-based integration and counterfactual inference to obtain batch-corrected, imputed gene expression counts across all slides. Principal component analysis of the top 3,000 variable genes revealed that tissue niche exerts a strong influence on the transcriptional profile of spots, with continuous variation in expression patterns as we navigate across niches (Fig~\ref{fig:cat}A). Spot-deconvolution analysis (shown in Figure~\ref{fig:hypxoia}G) shows specific localization of tumor states to specific spatial microenvironments. To assess whether our inferred networks capture this spatial heterogeneity in activity of tumor states, we computed activity scores for the Neftel modules across the four niches in the primary GBM network. Connectivity within the corresponding gene modules for each tumor state vary significantly across niches (Figure~\ref{fig:Fig-3 suppl}A), with specific localization of the MES2-like state in the PAN regions and NPC2-like state to the tumor leading edge(Figures \ref{fig:cat}B, \ref{fig:Fig-3 suppl}A) as reported in recent publications \cite{ravi2022spatially,greenwald2024integrative}.

\begin{figure}[h!]
    \centering
    \includegraphics[width=\textwidth]{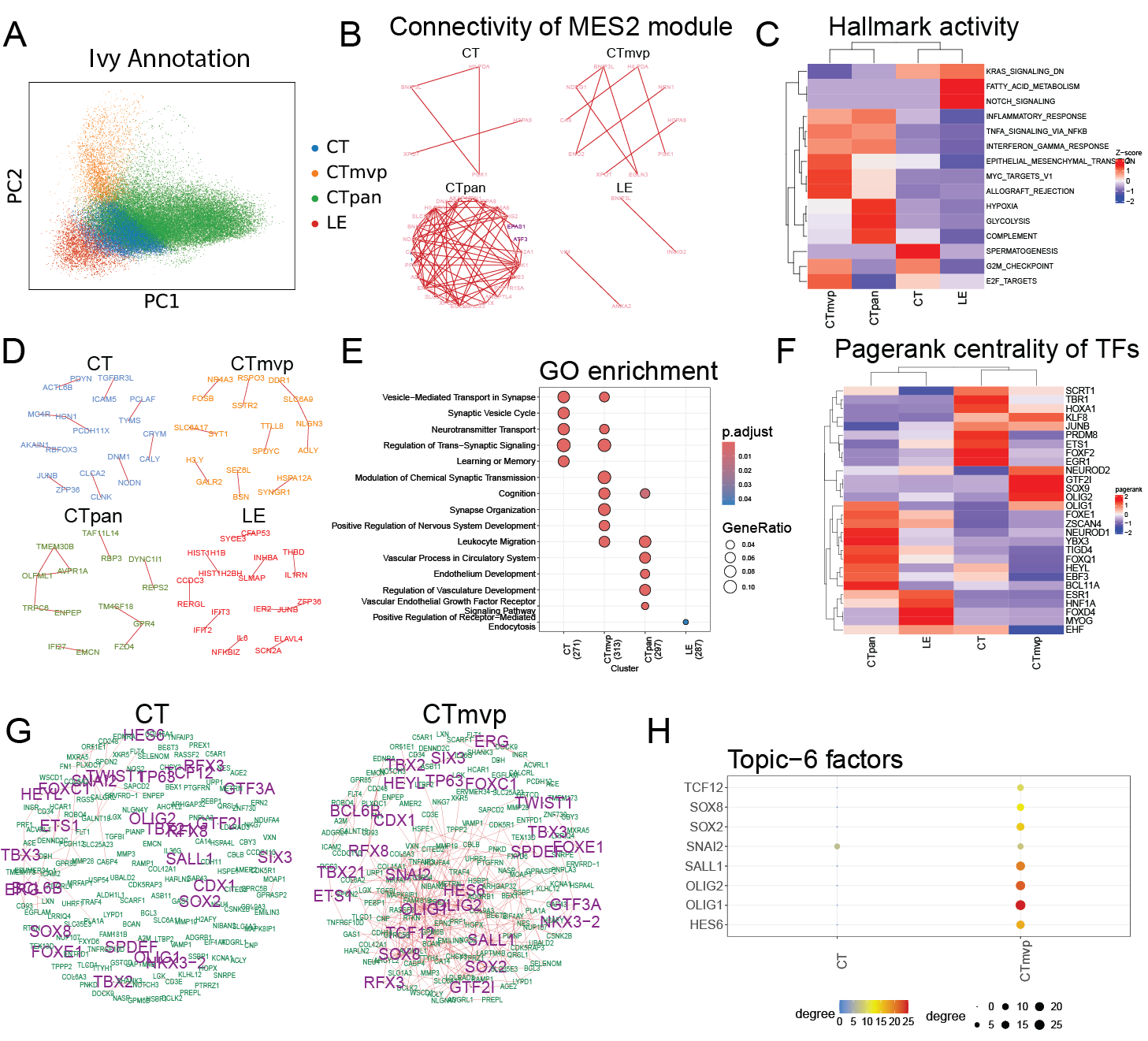}
    \caption{\textbf{Niche-specific gene regulatory network rewiring in recurrent glioblastoma.} (A) Principal component embedding of batch-corrected gene expression profiles reveals pronounced variation between anatomically defined tumor niches. (B) Connectivity analysis of the Neftel MES2 gene set in niche-specific regulatory graphs demonstrates significant enrichment of mesenchymal-like regulatory activity within the cellular tumor pan (CTpan) niche. (C) Variation in Hallmark pathway activity across primary tumor networks, highlighting distinct functional programs in each niche. (D) Top 10 differential regulatory edges unique to each niche in recurrent tumor networks, illustrating region- and recurrence-specific network rewiring. (E) Enrichment analysis of hub genes in differential networks identifies niche-adaptive regulators gained in recurrence. (F) Key TFs ranked by Pagerank centrality in differential networks, pinpointing master regulators underlying recurrent adaptation in each niche. (G) Topic modeling of perivascular (MVP) networks identifies a gene module associated with stemness, uniquely enriched in the MVP niche at recurrence. (H) Top TFs by node degree in the stemness-associated topic network, highlighting central regulators of therapy-resistant, stem-like cell states.}
    \label{fig:cat}
\end{figure}

\begin{figure}[h!]
    \centering
    \includegraphics[width=0.85\textwidth]{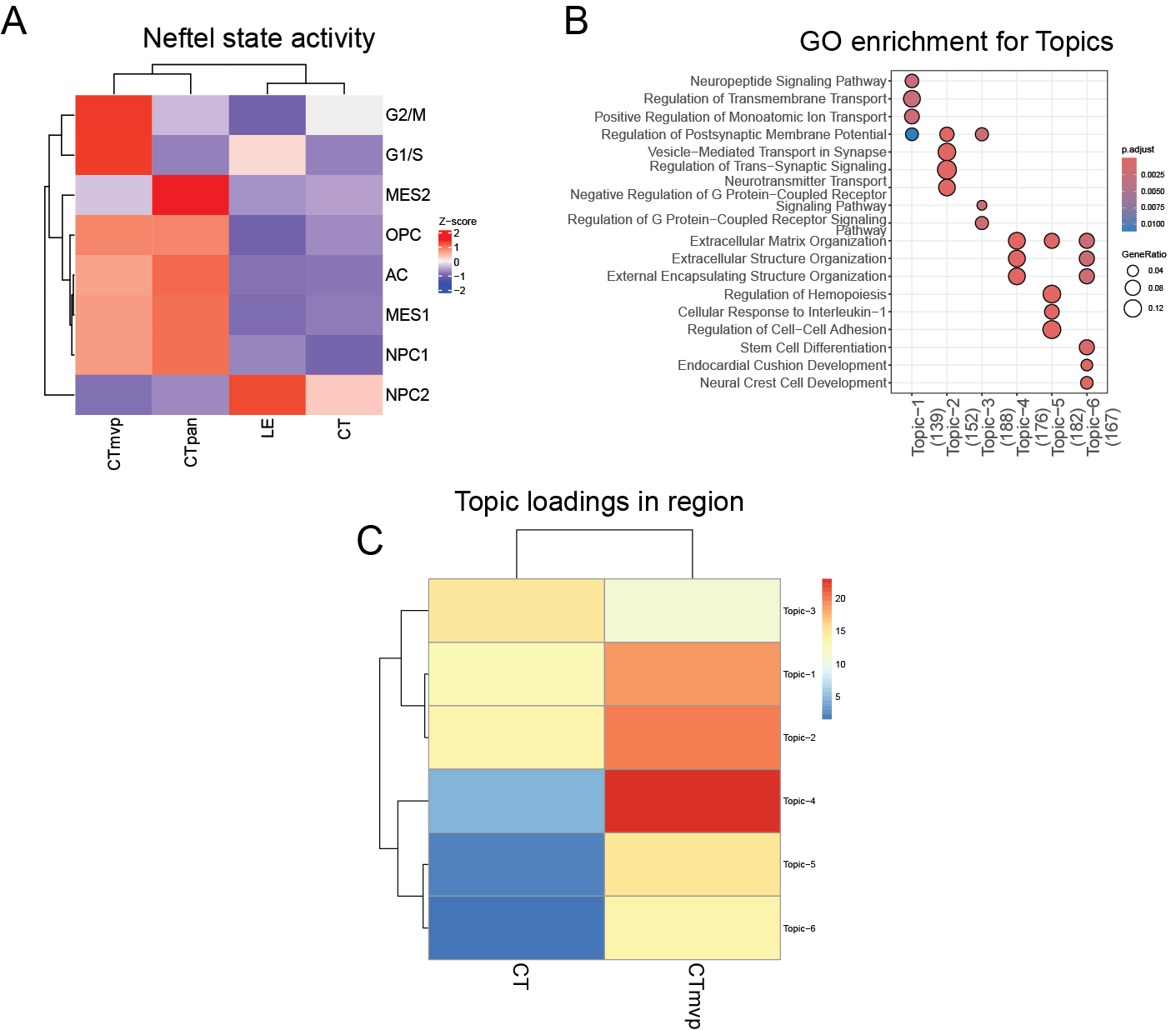}
    \caption{\textbf{Supplementary to Figure~\ref{fig:cat}: }(A) Relative enrichment scores for Neftel states in niche-specific networks in primary GBM. (B) GO enrichment for TF topics in peri-vascular (CTmvp + CT) niches. (C) Activity scores of topic modules in CTmvp vs CT networks.}
    \label{fig:Fig-3 suppl}
\end{figure}

To further evaluate the biological variability in our inferred networks, we quantified Hallmark gene set connectivity and visualized activity of the top five pathways per region (Figure~\ref{fig:cat}C). The perivascular niche (MVP) displayed strong connectivity for E2F targets and MYC signaling, pathways essential for glioma stem cell maintenance and survival via transcriptional and chromatin remodeling \cite{zhang2019chromatin, garnier2019glioblastoma}. The cellular tumor (CT) region showed robust cell cycle and EMT activity, indicative of highly proliferative and migratory cells \cite{jung2021tumor}. PAN regions were enriched for hypoxia and glycolysis, reflecting metabolic adaptation, while the leading edge (LE) was characterized by high Notch signaling, a pathway implicated in GSC survival, invasion, and therapy resistance \cite{wang2024mechanism}.

To identify regulatory interactions associated with recurrence, we focused on edges gained in the recurrent condition relative to primary, stratified by region (Figure~\ref{fig:cat}D). In the CT region, new interactions involved neuronal signal transduction and cell cycle regulators (e.g., \textit{PCLAF, TYMS}), while the MVP region gained edges related to cell adhesion, ECM remodeling, vascular development, and neurotransmitter release. PAN regions exhibited rewiring of genes involved in endothelial development, suggesting enhanced angiogenesis, and the LE in recurrent tumors showed increased connectivity among interleukin and interferon-response genes, indicating a more inflamed microenvironment. Enrichment analysis of hub genes in the differential networks (Figure~\ref{fig:cat}E) revealed that CT and MVP regions gained regulators of synaptic vesicle transport and neuronal function, with MVP uniquely enriched for neurodevelopmental pathways, potentially reflecting increased stemness at recurrence. PAN region hubs were associated with vascular growth, while LE hubs modulated specific signaling pathways. Pagerank centrality analysis of transcription factors highlighted region-specific master regulators (Figure~\ref{fig:cat}F): CT was dominated by developmental TFs (\textit{TBR1}, \textit{SCRT1}, \textit{HOXA1}, \textit{KLF8}, \textit{PRDM8}), MVP by stemness regulators (\textit{SOX9}, \textit{OLIG1/2}, \textit{GTF2I}), LE by \textit{EHF}, \textit{HNF1A}, and \textit{ESR1} (all linked to proliferation and poor prognosis \cite{ren2021identification, luo2024transcription}), and PAN by \textit{NEUROD1}, \textit{EBF3}, \textit{ZSCAN4}, \textit{BCL11A}, and \textit{FOXQ1}, which regulate neuronal differentiation, genomic stability, and EMT.

Given the similar enrichment profiles for genes in MVP and CT networks and the spatial proximity of these niches, we applied topic modeling to the differential graphs, identifying six gene modules (Figure \ref{fig:Fig-3 suppl}B). Topics related to neuronal signaling, cell adhesion, and ECM remodeling were differentially activated, with ECM-related modules significantly more active in the vascular niche. A stemness-associated topic (Topic-6) was also enriched in MVP (Fig~\ref{fig:Fig-3 suppl}C), suggesting a link to glioma stem cell programs \cite{brooks2017vascular}. The combined enrichment for ECM organization and stemness-related processes in the Topic-6 network suggests this module is specifically associated with glioma stem cells \cite{brooks2017vascular, wang2021tumor}. Fig~\ref{fig:cat}G shows the components of this topic and its significantly higher connectivity in the MVP region relative to the core tumor areas. Fig~\ref{fig:cat}H highlights the key TFs involved in this module and their higher node degree in the MVP-specific network, implicating these factors as potential drivers of increased stemness and invasiveness in treatment-resistant disease in this tumor niche. These findings are consistent with recent reports that the perivascular niche serves as a reservoir for therapy-resistant GSCs and is a major driver of tumor recurrence and adaptation \cite{wang2021tumor, ghaffari2021computational}.

The ability to resolve region-specific regulatory rewiring provides a powerful framework for dissecting the molecular mechanisms underlying tumor adaptation, progression, and treatment resistance. These insights not only deepen our understanding of the spatial biology of glioblastoma---but also inform the rational design of targeted interventions aimed at disrupting key adaptive programs within the most aggressive and treatment-refractory tumor compartments. \vspace{2mm}

\section{Discussion}

We introduce an $\ell_0$-penalized joint network inference algorithm for reconstructing gene regulatory networks (GRNs) from single-cell and spatial transcriptomics data. By directly enforcing network sparsity, our method avoids the indiscriminate edge shrinkage typical of $\ell_1$-based approaches, yielding more accurate and interpretable regulatory interactions \cite{srep20533, pmc8687351}. The framework is broadly applicable to populations structured as tree-structured hypergraphs, enabling the study of GRN dynamics along developmental, evolutionary, and physiological axes \cite{sciencedirect2022, srep20533}.

Applied to glioblastoma single-cell data, our unsupervised approach resolved distinct tumor subpopulations, including a recurrent, migratory population with elevated \textit{BACH1} activity-a key driver of invasiveness and therapy resistance \cite{science2024, pmc10519590, pmc8687351, pubmed34923423, nature2023, srep39743}. Topic modeling further delineated functional modules active in neural progenitor cells, revealing core regulators of population-specific phenotypes \cite{pmc10519590, science2024}. Spatial transcriptomics analysis captured continuous network rewiring along hypoxic gradients, identifying regulators and modules involved in angiogenesis and adaptation \cite{science2024, nature2023}. Our shared inference framework enables joint reconstruction of regional networks across disease states, facilitating the study of niche-specific network rewiring in recurrent disease. By stratifying tumors into physiological regions and inferring region-specific differential networks, we pinpoint hub genes and transcription factors uniquely rewired during recurrence within specific microenvironments. This includes the activation of stemness and immune evasion programs in perivascular and hypoxic niches---microenvironments increasingly recognized as reservoirs for therapy-resistant glioma stem cells and drivers of tumor progression \cite{frontiers2022, pmc5958355, pmc6357107, pmc3677798, nature2023}. Further dissection of these networks highlights master regulators of gene modules governing stem cell maintenance and therapy resistance in the perivascular niche. These results underscore the value of region- and condition-specific network inference for unraveling the molecular mechanisms of tumor adaptation and treatment resistance \cite{nature2023, pmc4831073, pmc5958355, pmc3677798, pmc10034917}.

Despite these advances, our work has two principal limitations. First, our reliance on imputed gene expression counts as model input-necessitated by the sparsity and limited sensitivity of current sequencing platforms-introduces a dependency on the performance of the imputation method, particularly scVI. The fidelity of the inferred networks is thus contingent on how well scVI captures the underlying latent structure, especially in large, atlas-scale datasets spanning multiple batches. Transitioning from Gaussian graphical models to frameworks that directly model count data, such as generalized linear models (GLMs) with appropriate penalization, may help address this limitation. Recent efforts leveraging GLMs for gene network inference \cite{chatrabgoun2025covariate} are promising in this regard.

The second limitation concerns the need for parameter tuning, specifically the selection of backward mapping thresholds when inferring networks across multiple populations. In large-scale applications, such as our GBM single-cell case study, we found that a single global threshold was insufficient and required cluster-specific adjustments. Although, in theory, a universal threshold with proper scaling should suffice, the high degree of correlation and non-independence among single-cell and spatial transcriptomics data complicates this process. This challenge appears to be intrinsic to the data structure, and alternative solutions remain elusive.

Looking ahead, we envision several avenues for extending this work. Integrating orthogonal epigenomic data, such as single-cell ATAC-seq from multiome experiments, could provide informative priors for network inference, thereby enhancing the recovery of biologically meaningful interactions \cite{baur2020data,Miraldi2019,kamimoto2023dissecting, bravo2023scenic+,Kim2023}. Such integration would enable a more comprehensive and mechanistically grounded analysis of regulatory networks. Additionally, we aim to explore the feasibility of in silico perturbation experiments using the inferred networks, with validation against CRISPR screens~\cite{kamimoto2023dissecting, Roohani2024,Gavriilidis2024}. This approach could substantially increase the translational impact of our methods, facilitating the identification of highly context-specific therapeutic targets for tumor cell states.

\section*{Acknowledgements}
S.F and A.B are partially supported by the NSF CAREER Award CCF-2337776, the NSF Award DMS-2152776, and the ONR Award N00014-22-1-2127. V.R and A.R are supported by the NCI R37CA214955-01A1 and the NSF Award DMS-2152776.

\section*{Competing interests}
A.R. serves as a member of Voxel Analytics, LLC, and serves as a consultant to Tempus, Telperian. He also serves as faculty advisor to TCS Ltd, and affiliate investigator at the Fred Hutch Cancer Center, and is a Satish Dhawan Visiting Professor at the Indian Institute of Science, Bangalore, India.

\bibliography{mybib.bib}    
\nolinenumbers

\end{document}